\documentclass[11pt]{amsart}
\usepackage{amssymb}
\usepackage{verbatim}

\begin{document}
\newtheorem{theorem}{Theorem}[section]
\newtheorem{proposition}[theorem]{Proposition}
\newtheorem{lemma}[theorem]{Lemma}
\newtheorem{corollary}[theorem]{Corollary}
\newtheorem{defn}[theorem]{Definition}
\newtheorem{conjecture}[theorem]{Conjecture}
\newtheorem{question}[theorem]{Question}
\newtheorem{problem}[theorem]{Problem}
\newtheorem{Remark}[theorem]{Remark}

\theoremstyle{definition}
\newtheorem{definition}{Definition}

\numberwithin{equation}{section}

\newcommand{\R}{\mathbb R}
\newcommand{\C}{\mathbb C}
\newcommand{\TT}{\mathbb T}
\newcommand{\Z}{\mathbb Z}
\newcommand{\HH}{\mathbb H}
\renewcommand{\^}{\widehat}
\newcommand{\vE}{\mathcal E} 
\newcommand{\vF}{\mathcal F}

\newcommand{\dist}{\operatorname{dist}}
\newcommand{\supp}{\operatorname{supp}}
\newcommand{\spec}{\operatorname{spec}}
\newcommand{\diam}{\operatorname{diam}}
\newcommand{\vol}{\operatorname{vol}}
\newcommand{\Ccap}{\operatorname{Cap}}
\newcommand{\Div}{\operatorname{div}}

\newcommand{\im}{\operatorname{Im}}
\newcommand{\re}{\operatorname{Re}}

\newcommand{\SL}{\operatorname{SL}}
\newcommand{\PSL}{\operatorname{PSL}}
\newcommand{\surface}{\Sigma}
\newcommand{\shape}{S}

\title{On the nodal sets of toral eigenfunctions}

\author{Jean Bourgain and Ze\'ev Rudnick}
\address{School of Mathematics, Institute for Advanced Study,
Princeton, NJ 08540 }
\email{bourgain@ias.edu}

\address{Raymond and Beverly Sackler School of Mathematical Sciences,
Tel Aviv University, Tel Aviv 69978, Israel 
}
\email{rudnick@post.tau.ac.il}
\date{\today}

\begin{abstract}
We study the nodal sets of eigenfunctions of the Laplacian
on the standard $d$-dimensional flat torus. The question we address
is: Can a fixed hypersurface lie on the nodal sets of eigenfunctions
with arbitrarily large eigenvalue? In dimension two, we show that
this happens only for segments of  closed geodesics. In higher dimensions,
certain cylindrical sets do lie on nodal sets corresponding to arbitrarily large eigenvalues.
Our main result is that this cannot happen for hypersurfaces with
nonzero Gauss-Kronecker curvature.

In dimension two, the result follows from a uniform lower bound for
the $L^2$-norm of the restriction of eigenfunctions to the curve,
proved in an earlier paper \cite{BR}. In high dimensions we currently
do not have this bound. Instead, we make use of the real-analytic
nature of the flat torus to study variations on this bound for
restrictions of eigenfunctions to suitable submanifolds in the complex
domain.
In all of our results, we need an arithmetic ingredient concerning the
cluster structure of lattice points on the sphere. We also present an
independent proof for the two-dimensional case relying on the
``abc-theorem'' in function fields.

\end{abstract}
\maketitle

\tableofcontents

\section{Introduction and statement of results}

Our goal in this paper is to study the nodal sets of high-frequency
eigenfunctions on the standard flat torus $\TT^d = \R^d/\Z^d$.
The eigenvalues of the Laplacian on $\TT^d$ are of the form
$4\pi^2 \lambda^2$, with $\lambda^2$ an integer, with corresponding
eigenfunctions which are trigonometric polynomials of the form
\begin{equation}\label{eigenfun on torus}
\varphi(x) = \sum_{\substack{|\xi| = \lambda\\ \xi\in \Z^d}} a_\xi e^{2\pi i \langle \xi, x \rangle}
\end{equation}
all of whose frequencies are integer points on the sphere $|x| =\lambda$. If $\lambda\neq 0$ then the  mean value $\int_{\TT^d} \varphi(x)dx=0$ vanishes. The nodal set is the locus of zeros $\{\varphi(x)=0\}$, which is a hypersurface
(codimension one) in $\TT^d$, necessarily real-analytic, possibly with
singularities.
We would like to study how the nodal sets vary when we increase the eigenvalue. It is known that for any real-analytic Riemannian manifold, the volume of the nodal sets is commensurable with $\lambda$ \cite{DF}.
In this paper, we address a different question: As $\lambda$ grows, can a fixed hypersurface lie on infinitely many nodal sets?

\subsection{Dimension $d=2$}

In the flat torus in two dimensions, we can have fixed curves where many eigenfunctions vanish, as in the case the line
$y=0$ which is on the nodal set of all the eigenfunction $\sin (2\pi
mx) \sin (2\pi ny)$. More generally, if $\surface$ is a {\em closed geodesic},
one can easily construct an infinite sequence of eigenvalues with eigenfunctions vanishing on $\surface$.
However this is essentially the only such possibility for the flat torus in two dimensions, where we can settle this problem
completely:
\begin{theorem}\label{thm:2 dim}
\label{nonvanishing thm}
 Let $\surface \subset \TT^2$ be a real-analytic curve. Then a necessary and
 sufficient condition that there are eigenfunctions $\varphi_\lambda$
 with arbitrarily large frequencies which vanish
on $\surface$ is that it be a segment of a closed geodesic.
\end{theorem}

Theorem~\ref{thm:2 dim} is an easy consequence of our uniform $L^2$-restriction
theorem \cite{BR} on the torus, which shows that for any smooth curve $\surface
\subset \TT^2$ with nowhere zero curvature,
there is some $\lambda_\surface>0$ and $C_\surface>0$ so that for all
eigenfunctions $\varphi_\lambda$ with $\lambda \geq \lambda_\surface$, we have
\begin{equation}\label{restriction lower bound}
  \int_\surface |\varphi_\lambda|^2  \geq C_\surface ||\varphi_\lambda||^2
\end{equation}
The restriction lower bound  \eqref{restriction lower bound} can be
used to show nonvanishing on curves (Theorem~\ref{thm:2 dim}) which are not
segments of closed geodesics:
Indeed, since the curve is real-analytic, if it is not flat it has
only finitely many flat points, hence by shrinking it we may assume
that it has nowhere-zero curvature, and then non-vanishing follows
from the lower bound \eqref{restriction lower bound}.
If $\surface$ is flat, but not a segment of a closed geodesic,
then it is a segment of an {\em unbounded} geodesic, in which case it is easy to check that
\underline{no} eigenfunction can vanish on it
(though we do not know the restriction bound \eqref{restriction lower bound} in this case, see \cite{Sarnak}).


We also have a completely different proof of Theorem~\ref{thm:2 dim}
using an algebraic argument, relying on the ``ABC theorem'' of
Brownawell-Masser \cite{BM} and Voloch \cite{Voloch} which we
present in Appendix~\ref{sec:abc}.

\subsection{Higher dimensions}

Suppose now that $\surface\subset \TT^d$ is a hypersurface (codimension one).
A special role is played by {\em flat} hypersurfaces, which on the torus are affine hyperplanes.
As in the two-dimensional case,  if $\surface$ is flat and  {\em closed}  (a closed totally geodesic hypersurface),
then it is contained in the nodal set of eigenfunctions with arbitrarily large
eigenvalues, e.g. if  $\surface=\{x\in \TT^d :\langle \xi, x\rangle = c\}$ for some $\xi\in \Z^d$, then it is part of the nodal set of the eigenfunctions $\varphi_n(x)=\sin  2\pi n (\langle \xi, x\rangle - c)$ for all $n\geq 1$.
However, in more than two dimensions, we do have non-flat hypersurfaces
(that is such that not all principal curvatures vanish)
contained in the nodal set of eigenfunctions with arbitrarily large
eigenvalues. For instance,  let $\varphi_0(x,y)$ be an eigenfunction
on the two-dimensional torus with eigenvalue $\lambda_0^2$, and
$S_0\subset \TT^2$ a curved segment contained in its nodal set.
For $n\geq 0$ let $\varphi_n(x,y,z) = \varphi_0(x,y)\cos 2\pi n z$, which is
an eigenfunction on $\TT^3$ with eigenvalue $\lambda_n^2=\lambda_0^2+n^2$, whose
nodal set contains the cylindrical set $\surface=S_0\times S^1$. Thus if $S_0$
is curved  then $\surface$ is not flat yet lies within the nodal set of all the
$\varphi_n$.
A similar construction works to show that there are  $\surface\subset
\TT^d$ which are cylindrical in the direction of any closed geodesic
for which there are eigenfunctions with arbitrarily large eigenvalues
vanishing on $\surface$.

So assume that $\surface$ has nowhere zero Gauss-Kronecker curvature, meaning
all principal curvatures are nowhere zero (see \S~\ref{sec:geometric}  for a
discussion).
We then show a higher-dimensional version of  Theorem~\ref{thm:2 dim}:
\begin{theorem}\label{thm:nonzero curvature}
Let $\surface\subset \TT^d$ be a real analytic (codimension one)
hypersurface, with nowhere-vanishing Gauss-Kronecker curvature.  Then there
is some $\lambda_\surface>0$ so that if $\lambda\geq \lambda_\surface$,
then $\surface$ cannot lie within the nodal set of any eigenfunction
$\varphi_\lambda$.
\end{theorem}

A key ingredient in this result deals with any hypersurface which is not flat.
As noted above, there are examples of such hypersurfaces contained in the nodal set of eigenfunctions with arbitrarily large eigenvalues. Our next result constrains the possible frequencies of such eigenfunctions, showing that the Fourier coefficients $a_\xi$ are negligible for frequencies $\xi$ whose directions $\xi/|\xi|$ lie outside a fixed cap on the sphere:
\begin{theorem}\label{thm:nonflat}
Let $\surface\subset \TT^d$ be a real analytic hypersurface which is not flat.
Then there is a cap $\Omega_\surface\subset S^{d-1}$ so that for all eigenfunctions
$\varphi_\lambda$ which vanish on $\surface$ we have
  \begin{equation}
    |a_\xi| \ll \frac {||\varphi_\lambda||_2}{\lambda^N}\,, \quad \forall N>1
  \end{equation}
for all $\xi\in \vE$ such that $\xi /|\xi |\in \Omega_\surface$.
\end{theorem}
Here and elsewhere we use the notation $f\ll g$ to mean that there is
some $c>0$ so that $f\leq cg$.

\subsection{About the proofs}
At this time we do not have an analogue of the uniform restriction theorem \eqref{restriction lower bound} in
arbitrary dimension (except dimension three \cite{BRdim3}). We refer to \cite{BGT, Hu} for upper bounds in a more general context.
We are also not able to use an analogue of the ``abc theorem'' as in \S~\ref{sec:abc}.
However, we retain the feature of passing to the complex domain, combined with some ideas from the $L^2$-restriction theorem,
to prove Theorem~\ref{thm:nonflat}  in arbitrary dimension.
The idea is that the eigenfunctions $\varphi$ are naturally
extended to holomorphic functions $\varphi^\C$ on the complexification $\C^d/\Z^d$ of $\TT^d$,
and since $\surface$ is real-analytic it too admits a complexification
$\surface^\C\subset \C^d/\Z^d$.
We then show in \S~\ref{sec:geometric} that there is a fixed cap of directions $\Omega_\surface\subset S^{d-1}$, and $\tau=\tau_\surface>0$ depending only on $\surface$,
so that for $v\in \Omega_\surface$,  there is a submanifold $\surface(v,\tau)\subset \surface^\C$
such that for all $Z\in \surface(v,\tau)$ the imaginary part  $\im Z=t v$ is parallel to $v$, and $\tau<t<2\tau$.
For frequencies $\xi$ for which $v:=-\xi/|\xi|\in \Omega_\surface$, we give  in \S~\ref{lower bound for integral}
a lower bound for the $L^2$-restriction of $\varphi^\C(Z)e^{-2\pi i \langle \xi,Z\rangle}$ to $\surface(v,\tau)$ of the form
\begin{equation}\label{eq:int mean square}
 \int_{\surface(v,\tau)} \left| \varphi^\C(Z)e^{-2\pi i \langle   \xi,Z\rangle}
\right|^2 d\mu(Z)   \gg |a_\xi|^2 + O(\frac 1{\lambda^N})
\end{equation}
where $d\mu$ is a smooth measure on $\surface(v,\tau)$. To compare
with the restriction theorem \eqref{restriction lower bound}, note
that in that case the RHS is $\sum_\xi|a_\xi|^2$.
That $\surface$ is not flat is used to guarantee decay of certain oscillatory integrals in \S~\ref{sec:oscillatory integral}.
On the other hand, if $\varphi$ vanishes on $\surface$ then its holomorphic extension $\varphi^\C$ will vanish on $\surface^\C$ and in particular the LHS of \eqref{eq:int mean square}  will vanish. This will prove Theorem~~\ref{thm:nonflat}.

To get vanishing of {\em all} Fourier coefficients and hence Theorem~\ref{thm:nonzero curvature},
we need all principal curvatures to be nonzero.
The extra argument needed to deduce it from Theorem~\ref{thm:nonflat}  is given in \S~\ref{sec:all freqs}.
In all of our results we need an arithmetic ingredient,
concerning the structure of lattice points on the sphere, which is given in \S~\ref{sec:cluster}.

\subsection{The sphere and Legendre polynomials}

One may investigate corresponding questions for the nodal sets of
eigenfunctions on other manifolds. However even in seemingly simple
situations the problem is as yet open. Consider the situation on the
two-dimensional sphere $S^2\subset \R^3$, where the Laplace-Beltrami
operator has eigenvalues $n(n+1)$ with the dimension of the
corresponding eigenspace $H_n$ being $2n+1$. We use spherical
coordinates: the colatitude $\theta\in [0,\pi]$ and longitude
$\phi\in [0,2\pi]$. In these coordinates, we may take as a basis of
the eigenspace $H_n$ the spherical harmonics
$$Y_n^m(\theta,\phi)=P_n^m(\cos \theta)e^{i m\phi}, \quad
-n\leq m\leq n
$$
where $P_n^m(x)$ are Legendre functions  (to get real valued
functions replace $e^{im\phi}$ by sine and cosine). In particular for
$m=0$ one gets the zonal spherical functions
$Y_n(\theta,\phi)=P_n(\cos \theta)$ where $P_n(x)$ are Legendre polynomials
\begin{equation}
 P_n(x) = \frac 1{2^n} \sum_{j=0}^{\lfloor n/2 \rfloor} (-1)^j \binom{n}{j}\binom{2n-2j}{n-2j} x^{n-2j}
\end{equation}
which are orthogonal polynomials on  the interval $[-1,1]$. The Legendre polynomial $P_n(x)$ has $n$ simple zeros $x_{n,j}\in[-1,1]$.

The nodal set of the zonal spherical harmonic
$Y_n^0$ is the union of  the parallels $\theta=\theta_{n,j}$,
$j=1,\dots,n$ where $x_{n,j} = \cos \theta_{n,j}$ are the zeros of the Legendre polynomial $P_n(x)$.
Since $P_n(-x) = (-1)^n P_n(x)$, for odd $n$
we have $P_n(0)=0$, and so we find that the zonal spherical harmonics
$Y_n^0$  vanish on the equator $\theta=\pi/2$ for odd $n$.

The equator $\theta=0$ was singled out by our choice of coordinates,
but by symmetry a corresponding construction works for all rotations
of the equator.
Thus every closed geodesic on the sphere lies within the nodal set of
eigenfunctions with arbitrarily large eigenvalues, as happens on the
flat torus.

A simple version of our results for the flat torus is to ask whether the other parallels (besides the equator)
lie on nodal sets of infinitely many eigenfunctions. As a special case, one can conjecture that a parallel (other than the equator) cannot lie within the nodal set of more than one zonal spherical harmonic.
This special case is equivalent to the conjecture of Stieltjes \cite{Stieltjes} that $P_m(x)$ and $P_n(x)$ have no common roots except $x=0$ when $m,n$ are both odd.
In fact, in the same letter \cite{Stieltjes}, Stieltjes put forth the stronger conjecture that
$P_{2n}(x)$ and $P_{2n+1}(x)/x$ are irreducible.
This was taken up by Holt \cite{Holt} in 1912, and by Schur and his
student Hildegard Ille \cite{Ille}, see also \cite{Wahab1, Wahab2}.
Around 1960 irreducibility was known for all $n\leq 500$ with a few exceptions
(which nowadays are easily checked by computer).

\bigskip

\noindent{\bf Acknowledgments:} We thank Aaron Levine for his comments
on Appendix~\ref{sec:abc}, and the anonymous referees for their
comments and suggestions. 
J.B. was supported in part by N.S.F. grants DMS 0808042 and DMS 0835373. Z.R. was supported by the Oswald Veblen Fund 
during his stay at the Institute for Advanced Study and by the Israel Science Foundation  (grant No. 1083/10).

\newpage
\section{Cluster structure of lattice points on the sphere}\label{sec:cluster}

For  $R\geq 1$, we denote by $\vE= \vE_R = \Z^d\cap RS^{d-1}$ the set
of lattice points on the sphere of radius $R$ (assuming $R^2$ is an integer):
$$ \vE_R:=\{\xi\in \Z^d: |\xi| = R\}$$
As is well known, the number of points in $\vE$ satisfies $\vE_R\ll R^\epsilon$ for all $\epsilon>0$ in dimension $d=2$, while in higher dimension $\#\vE_R$ grows polynomially.
Jarnik's theorem \cite{Jarnik} places constraints on location of lattice points in small caps:
\begin{theorem}[Jarnik's Theorem]\label{Jarnik's thm}
 There is some $c_d>0$ so
that all lattice points in the cap
\begin{equation*}
 \vE \cap \{|x|=R: |x-x_0|< c_d R^{\frac  1{d+1}} \}
\end{equation*}
lie on an affine hyperplane.
\end{theorem}

We will need more information about the ``cluster structure'' of the set $\vE$.
We define recursively  two sequences $c(d), \delta(d)$ with initial conditions
\begin{equation}
 \delta(2)<\frac 13, \quad c(2)=0
\end{equation}
and satisfying for $d\geq 3 $,
\begin{equation}\label{def of c(d)}
 c(d)  = 2\max \left( c(d-1),\frac {d}{\delta(d-1)} \right)
\end{equation}
\begin{equation}\label{def of delta(d)}
 \delta(d) = \frac 1{2(d+1)(1+c(d))}
\end{equation}

\begin{proposition}\label{Separation lemma}
Let $\vE\subseteq \Z^d \cap\{|x|=R\}$ be a subset of the set of lattice
  points on the sphere of radius $R$.
If $\rho<R^{\delta(d)}$ then:

a)  For any subset  $\vF \subset \vE$,
there is an overset $\vF \subseteq \tilde \vF \subset \vE$ satisfying
\begin{equation}\label{comparable diameter}
\diam(\tilde \vF) \leq \diam(\vF) + \rho^{1+c(d)}
\end{equation}
\begin{equation}\label{separation of F}
\dist(\tilde \vF,\vE\backslash \tilde \vF) > \rho
\end{equation}

b)   We may decompose $\vE = \coprod_\alpha \vE_\alpha$ into subsets
  satisfying
  \begin{equation}\label{Separation of clusters}
    \dist (\vE_\alpha,\vE_\beta) >\rho, \quad \alpha\neq \beta
  \end{equation}
  \begin{equation}
    \diam \vE_\alpha <\rho^{1+c(d)}
  \end{equation}
\end{proposition}

To prove Proposition~\ref{Separation lemma}, we will need:
\begin{lemma}\label{lem:chain}
If $1<\rho<R^{\delta(d)}$, and  $x_0,\dots, x_K\in \vE$ are distinct
elements satisfying
\begin{equation}
 \dist(x_i,x_{i+1}) \leq \rho, \quad i=0,\dots,K-1
\end{equation}
then $K \ll \rho^{c(d)}$.
\end{lemma}
\begin{proof}

We prove the claim by induction on the dimension $d$.

For $d=2$, we note that by Jarnik's theorem, there is some $c_2>0$ such that all lattice points in
an arc $\{|x|=R: |x-x_0|< c_2 R^{1/3} \}$ are co-linear, an in
particular there can be at most two of them. Thus if $\rho<\frac 12
c_2 R^{1/3}$ then we cannot have a chain $x_0,x_1,x_2$ with
$\dist(x_i,x_{i+1})<\rho$ since then we would have three lattice point in
cap of size $c_2R^{1/3}$.

Now let $d\geq 3$ and assume the contrary, that there is some chain $x_0,\dots x_K$ of length $K>\rho^{c(d)}$.
Let $K'=\lfloor \rho^{c(d)} \rfloor$ and consider the initial chain $C'=\{x_0,\dots x_{K'} \}$ of length $K'$.
The diameter of this chain is at most $\diam(C') \leq K'\rho \ll \rho^{c(d)+1} <R^{\delta(d)(1+c(d))}$ and hence
by \eqref{def of delta(d)}, $\diam(C')<R^{1/2(d+1)}=o(R^{1/(d+1)})$.
Therefore by Jarnik's theorem, this subchain is contained in some
hyperplane $H$, that is in the intersection of the sphere of radius
$R$ with the hyperplane $H$, which is a $(d-2)$-dimensional sphere of
some radius $R_1 \ll R$.

Thus we get a $\rho$-chain of length $K'$ in dimension $d-1$. There
are two possibilities:

1) If $\rho<R_1^{\delta(d-1)}$  then the inductive hypothesis allows
us to conclude
$K'<\rho^{c(d-1)} = o(\rho^{c(d)})$ by \eqref{def of c(d)} ,
contradicting $K'\approx \rho^{c(d)}$.

2) If $\rho >R_1^{\delta(d-1)}$ then we bound  the
number of lattice points on a $(d-2)$-dimensional sphere of radius
$R_1$ by $(1+2R_1)^d\ll R_1^d$ by replacing the sphere by a
$d$-dimensional cube containing the sphere (this is crude but uniform
with respect to the hyperplane $H$). Hence $K'\ll R_1^{d} \ll
\rho^{\frac{d}{\delta(d-1)} }$, which contradicts
$K'\approx \rho^{c(d)}$ by \eqref{def of c(d)}.
\end{proof}

%
%
%
%
%

We may now prove Proposition~\ref{Separation lemma}:
\begin{proof}
 We set  $\vF_0:=\vF$ and define
\begin{equation*}
 \vF_i:=\vF \cup \{x\in \vE: \dist(x,\vF_{i-1}) \leq \rho \}
\end{equation*}
to be the set of lattice points at distance less than $\rho$ from the
previous set. So we have an ascending sequence
$$ \vF_0=\vF \subseteq \vF_1 \subseteq \vF_2 \subseteq \dots$$
If $\vF_0,\vF_1,\dots \vF_k$ are all distinct then $k<\rho^{c(d)}$
since then we can form a chain $x_0,\dots x_k$ of {\em distinct}
elements $x_i \in \vF_i\backslash \vF_{i-1}$,  with
$\dist(x_i,x_{i+1})\leq \rho$.
Hence by Lemma~\ref{lem:chain}  we have $k<\rho^{c(d)}$.

Thus for some $0\leq K<\rho^{c(d)}$ we must have $\vF_K=\vF_{K+1}$.
Note that if  $\vF_{K+1} = \vF_K$ then $\vF_{K+j} = \vF_K$ for all
$j\geq 1$ and by definition, if $y\in \vE\backslash \vF_K$ then
$\dist(y,\vF_K) >\rho$. Thus taking $\tilde \vF:=\vF_K$ we get a set
which is well separated from its complement, that is \eqref{separation
  of F} holds, and for any $y\in \tilde \vF$ there is some $x\in \vF$
with $\dist(x,y)<K\rho<\rho^{1+c(d)}$, so that \eqref{comparable diameter} holds.

To prove the second part, we take some lattice point $x_1\in \vE$ and
let $\mathcal F = \vE\cap \mbox{Ball}(x_1,\rho^{1+c(d)})$.
Using the first part we find an overset
$\mathcal F \subseteq {\tilde\vF} \subseteq \vE$ satisfying
\eqref{separation of F} and \eqref{comparable diameter} and set
$\vE_1=\tilde \vF$, so that
$\diam \vE_1 \ll \rho^{1+c(d)}$ and $\dist(\vE_1,\vE\backslash
\vE_1)>\rho$. Now replace $\vE$ by $\vE\backslash \vE_1$ and continue
the process.
\end{proof}

\newpage
\section{Some geometric constructions}\label{sec:geometric}

\subsection{Background from differential geometry}
Let $\surface\subset \TT^d$ be a real-analytic hypersurface, which we assume
is non-singular. 
We consider a small parametric patch on $\surface$ which we may assume looks
like a graph, that is there is a real-analytic function
$f(x_1,\dots,x_{d-1})$ so that
$$\gamma(x) = (x, f(x)),\quad |x|<\delta$$
is a parametrization of $\surface$.

For each point $p\in \surface$, denote by $T_p\surface$ the tangent space to $\surface$ at $p$. On $\surface$ we have the frame field
\begin{equation}
 X_j := \frac{\partial \gamma}{\partial x_j} = (0,\dots ,\overbrace{1}^{j},\dots,0, \frac{\partial f}{\partial x_j})
\end{equation}
which gives at each point $p$ a basis of $T_p\surface$. The general tangent vector may be given as the linear combination
\begin{equation}\label{tangent vector}
v=\sum_{j=1}^{d-1}  w_j X_j=(\omega, \nabla f \cdot \omega),\quad \omega=(w_1,\dots w_{d-1})
\end{equation}


A choice of a unit normal field to the hypersurface $\surface$ at the point $p=(x, f(x))$ is given by
\begin{equation}\label{unit normal}
N_p:= \frac 1{\sqrt{ 1+|\nabla f(x)|^2}} (-\nabla f(x),1)
\end{equation}
The unit normal field defines  the Gauss map $N :\surface\to S^{d-1}$.
The {\em shape operator} for the hypersurface $\surface$, determined by the choice \eqref{unit normal} of unit normal,
is the linear map
\begin{equation}
\shape_p: T_p\surface\to T_p\surface,\quad v\mapsto -\nabla_v N_p
\end{equation}
that is $\shape_p$ is (minus) the derivative of the Gauss map.

The shape operator is self-adjoint:
\begin{equation}
\langle \shape_p(u), v \rangle = \langle u, \shape_p (v) \rangle
\end{equation}
and associated to it one has a symmetric bilinear form, the second fundamental form
\begin{equation}
 II_p(u,v) = \langle \shape_p(u),v \rangle
\end{equation}
The coefficients of the second fundamental form with respect to the frame field $\{X_j\}$ may be computed explicitly in terms of the derivatives
$$
X_{i,j} = \frac{\partial^2\gamma}{\partial x_i \partial x_j} = (\vec 0, \frac{\partial^2 f}{\partial x_i\partial x_j})
$$
as
\begin{equation}
 \langle \shape(X_i),X_j \rangle = \langle N, X_{i,j} \rangle =
\frac{\frac{\partial^2 f}{\partial x_i\partial x_j}}{ \sqrt{|\nabla f|^2 +1} }
\end{equation}

The eigenvalues of the shape operator are the {\em principal curvatures} of $\surface$,
and the determinant of $\shape$ is called the {\em Gauss-Kronecker curvature} of $\surface$.
The hypersurface $\surface$ is {\em flat}, i.e. is an affine hyper-plane, if and only if the unit normal $N_p$ is constant,
which happens if and only if all principal curvatures vanish, that is the shape operator is identically zero.

Given a unit tangent vector $u\in T_p\surface$, the {\em normal curvature} of $\surface$ at $p$ in the direction $u$ is defined as
$$k(u) = \langle \shape_p(u),u \rangle $$
If we cut the hypersurface $\surface$ by the plane spanned by $u$ and the unit normal $N_p$,
we get a curve whose tangent at $p$ is the vector $u$, and whose curvature at $p$ is $k(u)$.
For any nonzero tangent vector $v$ given as in \eqref{tangent vector}, the normal curvature in direction $v$ is
\begin{equation}
 k(v)=\langle \shape(\frac{v}{|v|}),\frac{v}{|v|}\rangle =
\frac 1{ \sqrt{|\nabla f|^2 +1} } \frac{\omega^T D_{x,x} f \omega}{|\omega|^2 + (\nabla f\cdot \omega)^2}
\end{equation}
where $D_{x,x} f = ( \frac{\partial^2 f}{\partial x_i\partial x_j})$ is the Hessian matrix of $f$.

Directions  for which the normal curvature vanishes are called
{\em asymptotic directions}. Thus a tangent vector $v$ as in \eqref{tangent vector} points in an asymptotic direction if and only if
\begin{equation}
 \omega^T D_{x,x} f \omega = 0
\end{equation}
The set of asymptotic directions at a point $p$ is called the asymptotic cone.
The hypersurface $\surface$ is   flat at the point $p$ if and only if every direction is asymptotic, that is
the asymptotic cone coincides with the whole tangent space.

\begin{lemma}\label{lem:directions3}
  Suppose $\surface$ is not flat. Then after shrinking $\surface$, we can find  a cap
$\Omega_\surface \subset S^{d-1}$ of directions $v =(\omega,w_d)$ so that:

i) There is a point $p=\gamma(x) \in \surface$ so that $v$ is tangent to $\surface$ at $p$, equivalently satisfies
\begin{equation}\label{v tangent at p}
 \nabla f(x) \cdot \omega = w_d
\end{equation}

ii)  The direction $v$ is not an asymptotic direction for all $p\in \surface$, that is for all $x$ and all $v=(\omega,w_d)\in \Omega_\surface$  we have
  \begin{equation}\label{nonvanishing at p}
    \omega^T D_{xx}f(x ) \omega \neq 0
  \end{equation}
\end{lemma}

\begin{proof}
Since $\surface$ is not flat, the Hessian $D_{x,x}f$ is not identically zero
(if it were, $f(x)=a+b\cdot x$ would be linear hence $\surface$ would be flat).
Then we may assume by further shrinking $\surface$ that in fact the Hessian matrix $D_{x,x} f(x) \neq 0$ is nonzero
for {\em all} $|x|<\delta$.

Since $D_{x,x} f(0)$ is not the zero matrix, the asymptotic cone
$$\{\omega\in \R^{d-1}: \omega^T D_{x,x} f(0)\omega = 0\}$$
has lower dimension and hence there is an open cone of directions for which $ \omega^T D_{x,x} f(0)\omega \neq 0$.
Moreover, by continuity of $x\mapsto D_{x,x}f(x)$, we get some $\delta>0$ and an open cone $\mathcal C$ so that
\begin{equation}
 \omega^T D_{x,x} f(x)\omega \neq 0,\quad \forall |x|<\delta, \quad \forall \omega\in \mathcal C
\end{equation}


Consider the map
\begin{equation}
\begin{split}
 V:\mbox{Ball}(\vec 0,\delta) \times \mathcal C &\to \R^d \\
(x,\omega)&\mapsto v=\sum_{j=1}^{d-1} w_j X_j = (\omega,\nabla f(x)\cdot \omega)
\end{split}
\end{equation}
We claim that the  map $V$ is a submersion, i.e. the Jacobian of $V$ has maximal rank for each $(x,\omega)$, hence the image of $V$ contains an open set of directions $\Omega_\surface\subset S^{d-1}$.

To see this, compute the Jacobian of $V$:
\begin{equation}
 D_{x,\omega} V = (\frac{\partial V}{\partial \omega_i},  \frac{\partial V}{\partial x_j}) =
\begin{pmatrix}
 I_{d-1} & 0_{d-1} \\ \nabla f(x) & D_{x,x}f(x) \cdot \omega
\end{pmatrix}
\end{equation}
Since $\omega^T D_{x,x} f(x)\omega \neq 0$ for all $x$ and $\omega\in \mathcal C$, hence $D_{x,x} f(x)\cdot \omega\neq \vec 0$,
it follows that the rank of $D_{x,\omega}V$ is $d$ as claimed.

Thus for each unit vector $v=(\omega, w_d)\in \Omega_\surface$, there is a point $p=\gamma(x) \in \surface$ so that $v$ is tangent to $\surface$ at $p$, equivalently is orthogonal to the normal, so satisfies
\begin{equation}
 \nabla f(x) \cdot \omega = w_d
\end{equation}
Moreover, for all such $x$ we have $\omega^T D_{x,x} f(x) \omega \neq 0$.
\end{proof}
Note that if $v$ satisfies \eqref{v tangent at p}, \eqref{nonvanishing at p} then so does $-v$.

\subsection{The complexification of $\surface$ and the submanifolds $\surface(v)$}

Since $f$ is real-analytic, there is a  holomorphic extension $F(z)$ of
$f$ to some neighborhood $\mathcal U \subset \C^{d-1}$.
This gives a holomorphic extension of the parametrization
\begin{equation}
\gamma^\C :z\in \mathcal U \mapsto (z,F(z))
\end{equation}
and we define the image
\begin{equation}
\surface^\C:=\{\gamma^\C (x+iy) = (z,F(z)),z\in \mathcal U\}
\end{equation}
to be the holomorphic extension of the surface $\surface$.

Let $  v  \in S^{d-1}$ be a unit vector in the cap guaranteed
by Lemma~\ref{lem:directions3}, so there is some $p = (x_0,f(x_0))\in \surface$ with
$  v\perp N_p$.
For such $v$, we will define a submanifold
$\surface(v)\subset \surface^\C$ so that
\begin{itemize}
 \item If $Z\in \surface(v)$ then $\im Z=t v$ is parallel to $v$.

\item $\surface(v)\cap \surface\subseteq \{p\in \surface: N_p\perp v \}$, i.e. at the real
points $p$ of $\surface(v)$, the normal vector $N_p$ is orthogonal to $v$, equivalently $v$ is tangent to
$\surface$ at $p$.

\end{itemize}
To do so, write $v=(\omega,w_d)$ and note that the vector-valued function
$\im\gamma^\C(x+it\omega)-t\vec v$ is
real-analytic in $t$ and vanishes at $t=0$. Hence we
can write
\begin{equation}
  \im \gamma^\C(x+it\omega)-t\vec v = t(\vec 0, h(x,t))
\end{equation}
with $h(x,t)=h_v(x,t)$ real analytic. We want to define $\surface(v)$ by the vanishing of $h(x,t)$. To do so, we need:
\begin{lemma}
For  $v=(\omega,w_d)\in \Omega_\surface$, let $p = (x_0, f(x_0))\in \surface$ be such that $v\in T_p\surface$.
Then
\begin{equation}
 h(x_0,0)=0
\end{equation}
and
\begin{equation}\label{grad of h}
      \nabla h(x_0,0) = (D_{x,x}f(x_0)\omega,0) \neq \vec 0
\end{equation}
\end{lemma}
\begin{proof}
Start at a point $p=(x_0, f(x_0))\in \surface$ with $N_p\perp v$.
To find the value of the function $h(x,t)$ at $t=0$, expand
\begin{equation}
 \im F(x,t\omega) = \im F(x,0) + t\nabla_y\im F(x,0)\cdot \omega + O(t^2)
\end{equation}
By the Cauchy-Riemann equations, $\nabla_y\im F(x,0) = \nabla_x \re F(x,0) = \nabla f(x)$
and since $\im F(x,0)=0$ we find
\begin{equation}
 h(x,0) = \lim_{t\to 0} \frac {\im F(x,t\omega)}{t} - w_d = \nabla
 f(x)\cdot \omega - w_d 
\end{equation}
Since the normal direction to the surface $\surface$ at $p$ is given by the vector $(-\nabla f(x_0),1)$, we find that
if $N_p\perp v$ then $h(x_0,0)=0$.

To show  \eqref{grad of h}, note that since $f(x)$ is real-analytic,
its holomorphic extension $F$ satisfies $ \overline{F(z)} = F(\bar
z)$, that is
$$\re F(\bar z) =\re F(z), \quad \im F(\bar z) = -\im F(z)$$
showing that $\im F(x+it\omega)$ is odd in $t$, hence $h(x,t)$ is
even in $t$. Therefore we have
$$ \frac{\partial h}{\partial t}(x,0) = 0 $$
Moreover since $\im F(x,t\omega)$ is odd in $t$,
$$ \im F(x,t\omega) = t (\nabla_y \im F)(x,0)\cdot \omega +O(t^3)$$
and hence
$$ \nabla_x \frac {\im F(x,t\omega)}{t} = D_{xy} \im F(x,0) \cdot
\omega +O(t^2)
$$
By Cauchy-Riemann, $D_{xy} \im F(x_0,0) = D_{xx}\re F(x_0,0) =
D_{xx}f(x_0)$ and hence we find
$$\nabla_x h(x_0) = D_{xx}f(x_0)\cdot \omega$$
proving \eqref{grad of h}.

By construction of the cap $\Omega_\surface$ in Lemma~\ref{lem:directions3}, we know that $\omega^T D_{x,x}f(x)\omega \neq 0$
for all $x$; in particular $D_{xx}(f)(x_0) \cdot \omega\neq \vec 0$ and therefore $\nabla h(x_0,0)\neq \vec 0$.
\end{proof}

Since $\nabla_x h(x_0,0)=D_{x,x}f(x_0)\omega \neq \vec 0$, we may use the Implicit Function Theorem to guarantee that there is
a neighborhood of $(x_0,0)$ in the $(x,t)$ domain where $\nabla_x h(x,t)\neq \vec 0$ and the condition $h(x,t)=0$ defines a smooth $(d-1)$-dimensional submanifold.  After shrinking $\surface$ and relabeling, we may then assume that for all $(x,t)$,
\begin{equation}
 \frac{\partial h}{\partial x_1}(x,t)\neq 0
\end{equation}
Using the Implicit Function Theorem, we can then write
\begin{equation}
 x_1 = x_1(t,\^x), \quad \^x:=(x_2,\dots, x_{d-1})
\end{equation}
We then define
\begin{equation}
\surface(v):=\{ \gamma^\C((x_1(t,\^x),\^x)+it\omega): |t|<\delta, |\^x|<\delta \}
\end{equation}
Note that since $v$ varies in a compact set, we may choose $\delta>0$
to work uniformly for all such $v$. Hence for $\tau\ll \delta$ , $\surface(v)$ contains
the set
\begin{equation}
  \surface(v,\tau):=\{Z=\gamma^\C((x_1(t,\^x),\^x)+it\omega)\in \surface(v):   \tau<t<2\tau, |\^x|<\delta \}
\end{equation}
so that $Z\in  \surface(v,\tau)$ implies that $\im Z=tv$, with $t\in
(\tau,2\tau)$.

We define a smooth measure $d\mu$ on $\surface(v,\tau)$ by taking a
smooth bump function $\psi(t,\^x)$, supported in $t\in [\tau,2\tau]$,
and setting
\begin{equation}\label{definition of mu}
 \int_{\surface(v,\tau)} g(Z)d\mu(Z):=
\int_{\tau<t<2\tau}\int_{|\^x|<\delta} g(\gamma^\C ((x_1(t,\^x),\^x)+it\omega)) \psi(t,\^x) dt d\^x
\end{equation}
We will restrict $\psi$ by requiring that its support is disjoint from
the lower-dimensional set of $(t,\^x)$ satisfying
\begin{equation}\label{restricting stationary pts a}
\begin{split}
\langle (\nabla_x \re F)(x,t\omega),\omega \rangle &=   w_d \\
\langle  (\nabla_y \re F)(x,t\omega),\omega \rangle &=0
\end{split}
\end{equation}
This condition will be used in \S~\ref{sec:oscillatory integral}
to ensure decay of an oscillatory integral.


\newpage
\section{Using complexification}

We start with an eigenfunction of the Laplacian with eigenvalue $\lambda^2$
\begin{equation}
  \varphi_\lambda(\vec X) = \sum_{\xi\in \vE} a_\xi e^{2\pi i\langle \xi,\vec
  X\rangle}, \quad \vec X \in \TT^d
\end{equation}
which we normalize by
\begin{equation}
 \sum_\xi |a_\xi|^2 = 1
\end{equation}
We want to show that if $\surface$ has nowhere zero Gauss-Kronecker curvature,
then for $\lambda>\lambda_\surface$, $\varphi_\lambda$ cannot vanish on the fixed hypersurface $\surface$.
We proceed to do so by showing initially that the Fourier coefficients $a_\xi$ are negligible for all frequencies whose directions $\xi/|\xi|$ lie in a cap $\Omega_\surface$ depending only on $\surface$. All that is required for this is that $\surface$ not be flat:
\begin{theorem}\label{lem:small a's}
  Assume $\surface\subset \TT^d$ is not flat. Then there is a cap $\Omega_\surface\subset S^{d-1}$ so that for all eigenfunctions
$\varphi_\lambda$ which vanish on $\surface$ we have
  \begin{equation}
    |a_\xi| \ll \frac {1}{\lambda^N}\,, \quad \forall N>1
  \end{equation}
for all $\xi\in \vE$ such that $-\xi /|\xi |\in \Omega_\surface$.
\end{theorem}

\subsection{The strategy}

Fix $\xi_0  \in \vE$ so that
\begin{equation}\label{choosing xi'}
  v_0  = -\frac { \xi_0}{|\xi_0|}  \in \Omega_\surface
\end{equation}
lies in the set of directions guaranteed by Lemma~\ref{lem:directions3}.

 We have a holomorphic extension $\varphi^\C(\vec Z)$ of $\varphi$ by replacing $\vec X$ by $\vec Z=\vec
X+i\vec Y\in \C^d$, which is a function on $\C^d/\Z^d$.
We will give a lower bound for the mean square of $\varphi^\C(Z)e^{-2\pi i\langle \xi_0,Z\rangle}$ restricted to the submanifold $\surface(v_0,\tau)$:
\begin{equation}\label{eq:lower bound for mean square}
 \int_{\surface(v_0,\tau)} \left| \varphi^\C(Z)e^{-2\pi i\langle \xi_0,Z\rangle}\right|^2 d\mu(Z) \gg |a_{\xi_0}|^2 + O(\frac 1{\lambda^N})
\end{equation}
where $d\mu$ is the smooth measure on $\surface(v_0)$ constructed in \eqref{definition of mu}.

On the other hand, vanishing of $\varphi$ on $\surface$ implies vanishing of
$\varphi^\C(\vec Z)$ on the holomorphic extension $\surface^\C$ of our surface $\surface$.
In particular, the mean square of $\varphi^\C(Z)e^{-2\pi i\langle \xi_0,Z\rangle}$ along $\surface(v_0,\tau)$ vanishes:
\begin{equation}
 \int_{\surface(v_0,\tau)} \left| \varphi^\C(Z)e^{-2\pi i\langle \xi_0,Z\rangle}\right|^2 d\mu(Z)=0
\end{equation}
and combining with the lower bound \eqref{eq:lower bound for mean square}, this will prove Theorem~\ref{lem:small a's}.

To prove the lower bound on the mean square \eqref{eq:lower bound for mean square}, we show that on $\surface(v_0,\tau)$
we may represent $ \varphi^\C(Z)e^{-2\pi i\langle \xi_0,Z\rangle}$ up to negligible error
by a sum over frequencies in a small cap:
$$
\varphi^\C(\vec  Z)e^{-2\pi i \langle \xi_0,\vec Z \rangle} =
\sum_{\vE'} a_\xi e^{2\pi i \langle \xi-\xi_0,Z\rangle}
+ O(\frac 1{\lambda^N}),\quad Z\in \surface(v_0,\tau)
$$
where the sum is over a certain set $\vE'$ of frequencies contained in a cap of size $\approx \sqrt{\lambda }\log \lambda$ around $\xi_0$.

On squaring out the sum in we will be faced with oscillatory integrals of the form
\begin{equation}\label{defA of J}
  J_{\xi,\xi'} := \int_{\surface(v_0,\tau)}  e^{ 2\pi i (\langle
    \xi-\xi_0, Z \rangle - \langle
    \xi'-\xi_0, \bar Z \rangle)}  d\mu(Z)
\end{equation}
We will bound these integrals by
\begin{equation}
 J_{\xi,\xi'} \ll \frac 1{|\xi-\xi'|^N},\quad \xi\neq \xi'\in \vE'
\end{equation}
Here the fact that $\surface$ is not flat is crucial.
Armed with this estimate, we will prove \eqref{eq:lower bound for mean square} by using the cluster structure
of the set of frequencies $\vE$ shown in \S~\ref{sec:cluster}.

\subsection{An oscillatory integral}\label{sec:oscillatory integral}
We want to bound the oscillatory integral $J_{\xi,\xi'}$ in 
\eqref{defA of J}, 
or writing out explicitly,
\begin{equation}\label{def of J}
  J_{\xi,\xi'} = \int e^{2\pi i |\xi-\xi'| \phi_u (t,\^x)} \mathcal A_{\xi,\xi'}(t,\^x) dt d\^x
\end{equation}
where we write $u=\frac{\xi-\xi'}{|\xi-\xi'|}$ and for any vector $u=(u_1,\dots,u_d)$
the phase function $\phi_u$ is defined by
\begin{equation}
 \phi_u(t,\^x) = \langle u, G(t,\^x) \rangle = \langle u, (x, \re F(x,t\omega)) \rangle
\end{equation}
$$
\^x=(x_2,\dots, x_{d-1}), \quad x=(x_1,\^x)
$$
$$
G(t,\^x)  = \re \gamma^\C(x+it\omega) = (x, \re F(x,t\omega))
$$
and with amplitude
\begin{equation}
  \mathcal A_{\xi,\xi'}(t,\^x) = e^{- 2\pi t (A(\xi) + A(\xi'))} \psi(t,\^x)
\end{equation}
\begin{equation}
 A(\xi) = \langle \xi-\xi_0, v_0 \rangle
\end{equation}
and $\psi(t,\^x)$ is a bump function. The region of integration in the
$(t,\^x)$ domain is
a small ball such that $ \tau <t<2\tau$.

\begin{lemma}\label{non-cylindrical lemma}
Let $v_0=(\omega,w_d)$ be as given in \eqref{choosing xi'}. 
Then for all unit vectors $u$ orthogonal to $v_0$,
the phase function $\phi_u(t,\^x)$ is non constant,
and the stationary points of $\phi_u$ lie on a subset of lower
dimension, which is independent of $u$, namely the points $(t,\^x)$
where
\begin{equation}\label{restricting stationary pts}
\begin{split}
\langle (\nabla_x \re F)(x,t\omega),\omega \rangle &=   w_d \\
\langle  (\nabla_y \re F)(x,t\omega),\omega \rangle &=0
\end{split}
\end{equation}
\end{lemma}
\begin{proof}
 Write out the phase function  explicitly as
\begin{equation}
 \phi_u(t,\^x) = u_1x_1(t,\^x) + \sum_{j=2}^{d-1} u_j x_j + u_d \re F(x_1,\^x, t\omega)
\end{equation}
Assume first that
$$\frac{\partial x_1}{\partial t}\neq 0$$
on the domain of integration.

We first dispose of the possibility that $u_d=0$. In that case,
\begin{equation}
 \phi_u(t,\^x) = u_1x_1(t,\^x) + \sum_{j=2}^{d-1} u_j x_j
\end{equation}
 At a stationary point,
\begin{equation}
0= \frac{\partial \phi_u}{\partial t} = u_1\frac{\partial x_1}{\partial t}
\end{equation}
and since $\frac{\partial x_1}{\partial t}\neq 0$ on the support of $\psi$, we find $u_1=0$ so that $u=(0,u_2,\dots, u_{d-1},0)$ and
$\phi_u = \sum_{j=2}^{d-1} u_j x_j$ is linear, $\nabla \phi_u = u$ and $\phi_u$ has no stationary points.

Assume from now that $u_d\neq 0$. Consider the differential operator
\begin{equation}
  \mathcal L  = \frac{w_1}{\frac{\partial x_1}{\partial t}}
  \frac{\partial}{\partial t}  + \sum_{j=2}^{d-1}
  w_j(\frac{\partial}{\partial x_j} - \frac{\frac{\partial
      x_1}{\partial x_j}}{\frac{\partial x_1}{\partial t}}
  \frac{\partial}{\partial t})
= A\frac{\partial}{\partial t} + \sum_{j=2}^{d-1} w_j
  \frac{\partial}{\partial x_j}
\end{equation}
with
$$\quad A = \frac 1{\frac{\partial x_1}{\partial t}}
(w_1-\sum_{j=2}^{d-1} w_j \frac{\partial x_1}{\partial x_j})
$$
A calculation using the chain rule shows that for a function of the
form $H(x,t\omega)$ we have
\begin{equation}
  \mathcal L\{H(x,t\omega)\}  = \langle \nabla_xH,\omega \rangle +
A  \langle \nabla_y H,\omega \rangle
\end{equation}
Applying $\mathcal L$ to the phase function $\phi_u$, using $\mathcal
L x_j = w_j$, gives
\begin{equation}
  \mathcal L \phi_u = \sum_{j=1}^{d-1} u_jw_j + u_d(\langle \nabla_x \re
  F,\omega \rangle  + A\langle \nabla_y \re  F,\omega \rangle)
\end{equation}
Hence at a stationary point, where $\mathcal L \phi_u = 0$, we find on
using the orthogonality of $u$ and $v_0$, that
\begin{equation}
  u_d(\langle \nabla_x \re  F,\omega \rangle
+ A(\langle \nabla_y \re  F,\omega \rangle) = - \sum_{j=1}^{d-1}
u_jw_j = u_dw_d
\end{equation}
and since $u_d\neq 0$ we get
\begin{equation}\label{eq:7.14}
  \langle \nabla_x \re  F,\omega \rangle
+ A \langle \nabla_y \re  F,\omega \rangle = w_d
\end{equation}

Likewise, applying $\mathcal L$ to the relation
$\im F(x,t\omega) =tw_d$ we get on using $\mathcal L t=A$ that
\begin{equation}
  \langle \nabla_x \im F,\omega \rangle + A \langle \nabla_y \im
  F,\omega \rangle  = Aw_d
\end{equation}
Applying the Cauchy-Riemann equations $\nabla_x \im F =-\nabla_y \re
F$, $\nabla_y\im F = \nabla_x \re F$ gives
\begin{equation}\label{eq:7.15}
  A\langle \nabla_x\re F,\omega \rangle -\langle \nabla_y\re F,\omega
  \rangle = Aw_d
\end{equation}
The unique solution of the system \eqref{eq:7.14}, \eqref{eq:7.15} is then
\begin{equation}\label{eq:7.16}
  \langle \nabla_x\re F,\omega \rangle =w_d\,,
\quad \langle \nabla_y\re F,\omega  \rangle =0
\end{equation}

Now assume that $\phi_u$ is constant, so that \eqref{eq:7.16} holds
for all $t,\^x$. Then we may apply the differential operator $\mathcal
L$ to  \eqref{eq:7.16} to get
\begin{equation}\label{eq:7.17}
  \mathcal L  \langle \nabla_x\re F,\omega \rangle = \omega^T
  \nabla_{x,x} \re F \omega + A \omega^T \nabla_{x,y} \re F \omega =0
\end{equation}
and
\begin{equation}
   \mathcal L  \langle \nabla_y\re F,\omega \rangle =
\omega^T \nabla_{x,y} \re F \omega  + A \omega^T \nabla_{y,y} \re F
\omega = 0
\end{equation}
By the Cauchy-Riemann equations, $\nabla_{y,y} \re F = -\nabla_{x,x}
\re F$ and we find
\begin{equation}\label{eq:7.19}
  \omega^T \nabla_{x,y} \re F \omega  - A \omega^T \nabla_{x,x} \re F\omega = 0
\end{equation}

The unique solution of the system \eqref{eq:7.17}, \eqref{eq:7.19} is
\begin{equation}
  \omega^T \nabla_{x,x} \re F\omega = 0= \omega^T \nabla_{x,y} \re F \omega
\end{equation}
which assuming that $\phi_u$ is constant on $\surface(v_0)$, holds throughout
$\surface(v_0)$ and in particular on the real locus  $t=0$, where $\re F(x,0) =
f(x)$, where we find
\begin{equation}
  \omega^T \nabla_{x,x} f(x) \omega = 0, \quad \forall x\in \surface(v_0)\cap \surface
\end{equation}
This contradicts  \eqref{nonvanishing at p}, that is that $v_0$ is not an asymptotic direction at any point on $\surface$.
Thus $\phi_u$ is non-constant, and being real-analytic its stationary
points (where $\nabla \phi_u = \vec 0$) lie on a lower-dimensional
subset $\mbox{Crit}(u)$ of $\surface(v_0)$.

Since \eqref{restricting stationary pts}  holds at the stationary
points, which is independent of $u$,
$\mbox{Crit}(u)$ are confined to lie inside a lower-dimensional subset
which is independent of $u$.

Next consider the case $\frac{\partial x_1}{\partial t}\equiv 0$ on the support of $\psi$. We claim that still
\eqref{restricting stationary pts} holds.

We first dispose of the case $u_d=0$ when $\phi_u = u_1x_1+\sum_{j=2}^{d-1} u_j x_j$ with $u\cdot w= \sum_{j=1}^{d-1}u_jw_j=0$.
(We leave the case $d=2$ as an exercise). If $u_1=0$ then $\phi_u$ is a non-zero linear function and has no stationary points. Otherwise, at a stationary point,
\begin{equation}
 0= \frac{\partial \phi_u}{\partial x_j} = u_1 \frac{\partial x_1}{\partial x_j} +u_j
\end{equation}
Differentiating the relation $\im F(x,t\omega) = tw_d$ with respect to $x_j$ gives
\begin{equation}
 0=\frac{\partial x_1}{\partial x_j}\frac{\partial \im F}{\partial x_1} + \frac{\partial \im F}{\partial x_j} = -
\frac{\partial x_1}{\partial x_j}\frac{\partial \re F}{\partial y_1} - \frac{\partial \re F}{\partial y_j}
\end{equation}
by the Cauchy-Riemann equations. Hence
\begin{equation*}
 \begin{split}
  \langle \nabla_y \re F, \omega \rangle  &= w_1 \frac{\partial \re F}{\partial y_1} + \sum_{j=2}^{d-1} w_j \frac{\partial \re F}{\partial y_j} \\
&=  w_1 \frac{\partial \re F}{\partial y_1} -\sum_{j=2}^{d-1} w_j \frac{\partial x_1}{\partial x_j}\frac{\partial \re F}{\partial y_1}  \\&=  \frac{\partial \re F}{\partial y_1}(w_1 + \sum_{j=2}^{d-1} w_j\frac{u_j}{u_1})=0
  \end{split}
\end{equation*}
since by orthogonality of $u$ and $v_0$ and vanishing of $u_d$, we have $w_1 + \sum_{j=2}^{d-1} w_j\frac{u_j}{u_1}=0$.

Assume now that $u_d\neq 0$. Then at a stationary point,
\begin{equation}
0=\frac{\partial \phi_u}{\partial t} =u_d\langle \nabla_y \re F,\omega \rangle
\end{equation}
and $u_d\neq 0$ implies $\langle \nabla_y \re F,\omega \rangle=0$.
Differentiating the relation $\im F(x,t\omega) = tw_d$ with respect to $t$, using independence of $x_1$ relative to $t$, gives
\begin{equation}
 w_d = \langle \nabla_y \im F,\omega \rangle = \langle \nabla_x \re F,\omega \rangle
\end{equation}
by the Cauchy-Riemann equations. Thus in all cases \eqref{restricting stationary pts} hold at a stationary point.

Now differentiate the relation $\langle \nabla_y \re F(x,t\omega),\omega \rangle=0$ with respect to $t$,
keeping in mind that $x_1$ is independent of $t$, to get $ \omega^T \nabla_{y,y} \re F \omega = 0$
and using the Cauchy-Riemann equation we get $\omega^T \nabla_{x,x} \re F \omega =0$.
Specializing to the real locus $t=0$ again gives a contradiction.
\end{proof}

\begin{lemma}\label{lem:oscillatory integral}
Let $v_0$ be as in \eqref{choosing xi'}, and set $D=(\log \lambda)^2$. 
Then for all  $\xi\neq\xi'$ lying in a cap of size $\sqrt{\lambda D}$
around $\lambda v_0$ we have
  \begin{equation}
     |J_{\xi,\xi'} | \ll  \frac {1} {|\xi-\xi'|^r}  
\,,\quad \forall r\geq 1
  \end{equation}
\end{lemma}

\begin{proof}
Write
$$
J_{\xi,\xi'}= \int e^{2\pi i|\xi-\xi'| \Phi_{\xi,\xi'}(t,\^x)}
\mathcal A_{\xi,\xi'}(t,\^x) dtd\^xdx
$$
where
\begin{equation}
\Phi_{\xi,\xi'}(t,\^x):= \langle \frac{\xi-\xi'}{|\xi-\xi'|}, G(t,\^x) \rangle
\end{equation}

We claim that there is some $C>0$ for which for all  $\xi\neq\xi'$ in our
cap, the phase functions satisfy
\begin{equation}\label{grad Phi}
 ||\nabla \Phi_{\xi,\xi'}(t,\^x) || \geq C
\end{equation}
Indeed, decompose $\xi-\xi'$ into components along  $v_0$ and
orthogonal to it:
$$
\xi-\xi' = ku + \langle \xi-\xi' ,v_0 \rangle v_0\,, \quad u\perp v_0, \quad ||u||=1
$$
Since $\xi,\xi'$ lie in a cap of size $\sqrt{\lambda D}$ on the sphere
of radius $\lambda$,
the difference $\xi-\xi'$ is almost orthogonal to $v_0$ and 
we claim that 
\begin{equation}
 \left| \langle \frac{ \xi-\xi'}{|\xi-\xi'|} ,v_0\rangle \right| \ll
 \frac{\sqrt{\lambda D}}{\lambda } = o(1)
\end{equation}
Indeed, writing $\xi=\xi_0+\eta$, $\xi'=\xi_0+\eta'$, with
$|\eta|,|\eta'|\leq \sqrt{\lambda D}$ we get 
$$ 
2\langle \eta,\xi_0 \rangle +|\eta|^2 = 0 = 2\langle \eta',\xi_0
\rangle +|\eta'|^2
$$
and since $\xi_0=-\lambda v_0$, 
$$
\left| \langle \frac {\xi-\xi'}{|\xi-\xi'|}, v_0 \rangle \right| = 
\left| \langle \frac {\eta-\eta'}{|\eta-\eta'|}, v_0 \rangle \right|
= 
\left | \frac{|\eta|^2 - |\eta'|^2}{2\lambda |\eta-\eta'|}\right| \leq
\frac{|\eta|+|\eta'|}{2\lambda} = \frac{\sqrt{\lambda D}}{\lambda}
$$
as claimed. Likewise we have 
\begin{equation}
  A(\xi),A(\xi') \leq \frac D2
\end{equation}
because 
$$
A(\xi) = \langle \eta,v_0 \rangle  =  - \langle \eta, \frac
{\xi_0}{\lambda} \rangle = \frac{|\eta|^2}{2\lambda} \leq \frac D2
$$
Therefore we find
\begin{equation}
 |k| \sim |\xi-\xi'|
\end{equation}
Thus
$$
\Phi_{\xi,\xi'}(t,\^x) = \frac{k}{|\xi-\xi'|}  \phi_u(t,\^x) +
o(1)\langle v_0, G(t,\^x) \rangle
$$

By our choice \eqref{restricting stationary pts a} of $\psi$ and as
a consequence of Lemma~\ref{non-cylindrical lemma},
we know that $\phi_u$ has no critical point in $\supp \psi$ for all
$u\perp v_0$ and so there is some  $C>0$ so that
$|\nabla \phi_u(t,\^x)| >2C$
for all $u\perp v_0$ and all $(t,\^x)\in \supp \psi$. Therefore (recalling
that $|k|\sim |\xi-\xi'|$) for $\lambda \gg 1$ we have
$$
|\nabla \Phi_{\xi,\xi'}(t,\^x)| >C
$$
as claimed.

Integrating by parts we get that
\begin{equation}
  J_{\xi,\xi'} \ll ||\mathcal A_{\xi,\xi'} ||_{C^r} \frac 1{|\xi-\xi'|^r} \,
  ,\quad \forall r\geq 1
\end{equation}
Since $0\leq A(\xi),A(\xi')\leq D$, and $\tau\leq t\leq 2\tau$ on $\supp \psi$,
we may bound the $C^r$-norm  of the amplitude function $\mathcal A_{\xi,\xi'} = e^{-2\pi t(A(\xi) + A(\xi'))}\psi(t,\^x)$ by
\begin{equation}
 ||\mathcal A_{\xi,\xi'} ||_{C^r}  \ll_\psi  (1+A(\xi)+A(\xi'))^{r+1} e^{-2\pi \tau (A(\xi)+A(\xi'))}
=O(1)
\end{equation}
(the implied constant depends only on $\psi$, $\tau$ and $r$, not on $\xi,\xi'$), which gives the required estimate.
\end{proof}

\newpage
\section{A lower bound for the mean-square of $\varphi^\C(Z)e^{-2\pi i  \langle \xi_0,Z \rangle}$}
\label{lower bound for integral}


\subsection{Representing $\varphi^\C$ on $\surface(v_0,\tau)$ by a short sum}
We show that for $Z\in \surface(v_0,\tau)$ 
we may represent $\varphi^\C(\vec  Z)e^{-2\pi i \langle \xi_0,\vec Z \rangle}$ by the part of its Fourier expansion
whose frequencies lie in a small cap around $\xi_0$, up to a negligible error.

For $\xi\in \vE$, set
\begin{equation}
A(\xi)   = \langle \xi-\xi_0,v_0 \rangle
\end{equation}
Observe that since all vectors $\xi\in \vE$ lie on a sphere, and thus
no  two vectors can lie on the same positive ray, we have
$\langle \xi,\xi_0 \rangle < \langle \xi_0,\xi_0 \rangle$
for all vectors $\xi\neq \xi_0$ and hence
\begin{equation}
  A(\xi) >0, \quad \xi\neq \xi_0 \,,\quad  A(\xi_0) = 0
\end{equation}

Let
\begin{equation}
 D\approx (\log \lambda)^2
\end{equation}
 The set
\begin{equation}
  \vE':= \{\xi\in \vE: A(\xi) <D \}
\end{equation}
is contained in a cap of size $\approx \sqrt{\lambda D}$ centered at $\xi_0$.
Note that for $\xi$ in this cap, $\xi-\xi_0$ is almost perpendicular to $v_0$.

\begin{lemma}\label{lem:short sum}
 For $\vec Z=\vec X+i\vec Y\in \surface(v_0,\tau)$, we have
\begin{equation}\label{rep phi by short sum}
 \varphi^\C(\vec  Z)e^{-2\pi i \langle \xi_0,\vec Z \rangle} = \sum_{A(\xi)\leq D} a_\xi
e^{2\pi i \langle \xi-\xi_0, \vec X\rangle }
e^{-2\pi t A(\xi)} + O(\frac 1{\lambda^N})
\end{equation}
for all $N\geq 1$.
\end{lemma}
\begin{proof}
We define a subset $T(v_0) \subset \C^d/\Z^d$ by
\begin{equation}
 T(v_0):=\{\vec Z\in \C^d: \im \vec Z = |\im \vec Z| v_0 \}
\end{equation}
that is the complex vectors whose imaginary parts point along the ray in the direction of $v_0$.
Restricting $\varphi^\C$ to $T(v_0)$, we have
\begin{equation}
 \varphi^\C(\vec Z) e^{-2\pi i \langle \xi_0, \vec Z \rangle} = \sum_\xi a_\xi e^{2\pi i \langle \xi-\xi_0, \vec X\rangle }
e^{-2\pi t A(\xi)}\,, \quad t:=|\im \vec Z|
\end{equation}

Now restrict $\vec Z$ further by assuming that it lies in the set
$$ T(v_0;\tau):=\{\vec Z\in T(v_0):\tau <|\im \vec Z| <2\tau\}$$
Then for  $\vec Z \in T(v_0;\tau)$ we have
\begin{multline}
 |\sum_{A(\xi)>D} a_\xi e^{2\pi i \langle \xi-\xi_0,\vec Z\rangle } |\leq
\sum_{A(\xi)>D} |a_\xi|e^{-2\pi |\im \vec Z|D} \\
\ll (\#\vE)^{1/2}e^{-2\pi \tau D} \ll \frac 1{\lambda^N}\,, \quad \forall N>1
\end{multline}
using $\sum_\xi|a_\xi|^2=1$,  $\#\vE\ll \lambda^{d-2+\epsilon}$ and $A(\xi)\geq D$, $|\im \vec Z|>\tau$.

Hence for $\vec Z\in T(v_0;\tau)$ we have
\begin{equation}
 \varphi^\C(\vec  Z)e^{-2\pi i \langle \xi_0,\vec Z \rangle} = \sum_{A(\xi)\leq D} a_\xi
e^{2\pi i \langle \xi-\xi_0, \vec X\rangle }
e^{-2\pi t A(\xi)} + O(\frac 1{\lambda^N})
\end{equation}
In particular, since $\surface(v_0,\tau)\subset T(v_0,\tau)$ we proved \eqref{rep phi by short sum}.
\end{proof}

\subsection{Proof of the lower bound on the mean square}
We now want to prove the lower bound \eqref{eq:lower bound for mean square}
for the mean square, namely that
$$
\int_{\surface(v_0,\tau)}| \varphi^\C(Z)e^{-2\pi i  \langle \xi_0,Z
  \rangle} |^2 d\mu(Z) \gg |a_{\xi_0}|^2 + O(\frac 1{\lambda^N})
$$
Hence what we need follows from the following Lemma, once we recall
that $A(\xi)\geq 0$ for all $\xi\in \vE'$ and that $A(\xi_0)=0$:
\begin{proposition}\label{lem:lower bd for short mean square}
There is some $C>0$ so that
\begin{equation}
\int_{\surface(v_0,\tau)} \left| \sum_{\xi\in \vE'} a_\xi e^{2\pi i \langle \xi-\xi_0,Z \rangle} \right|^2 d\mu(Z)
\geq C \sum_{\xi\in \vE'} |a_\xi|^2 e^{-8\pi \tau A(\xi)}  +O(\frac 1{\lambda^N})
\end{equation}
for all $N>1$. 
\end{proposition}
\begin{proof}
We note that for $Z\in \surface(v_0,\tau)$,
$$ \langle \xi-\xi_0,Z \rangle = \langle \xi-\xi_0, G(t,\^x) \rangle +
it A(\xi) $$
where
\begin{equation}
 G(t,\^x) = \re \gamma^\C((x_1,\^x)+it\omega) =
(x_1,\^x, \re F((x_1,\^x)+it\omega))
\end{equation}
and so we need to show that
\begin{equation}\label{inductive assumption}
\int \left| \sum_{\xi\in \vE'} a_\xi e^{-2\pi tA(\xi)}
e^{2\pi i \langle \xi-\xi_0, G(t,\^x) \rangle }  \right|^2 d\mu(t,\^x)
\geq C \sum_{\xi\in \vE'} |a_\xi|^2 e^{-8\pi \tau A(\xi)}  +O(\frac 1{\lambda^N})
\end{equation}
for all $N>1$, where $  d\mu(t,\^x) = \psi(t,\^x)dtd\^x$.

Let $d(\vE')$ be the minimal dimension of an affine hyperplane which
contains all the frequencies $\vE'$, so $ d(\vE') \leq d$. We show by
induction on $d'=d(\vE')$ that \eqref{inductive assumption} holds.

The case $d'=1$: This means all frequencies lie on a line, and hence (since they lie on a sphere) there are at most two of them. If there is exactly one frequency the claim is clear, so we need to treat the case $\vE'=\{\xi,\xi'\}$ consists of two distinct frequencies.
That is we want to show that
\begin{multline}\label{base case d=1}
 \int \left| a_\xi e^{2\pi i \langle \xi-\xi_0,G(t,\^x) \rangle }
 e^{-2\pi tA(\xi)} +
a_{\xi'}e^{2\pi i \langle \xi'-\xi_0,G(t,\^x) \rangle } e^{-2\pi  tA(\xi')} \right|^2
d\mu \\
\geq C( |a_\xi|^2e^{-8\pi \tau A(\xi)}  + |a_{\xi'}|^2 e^{-8\pi \tau A(\xi')} )
\end{multline}

Write
$$
\mathcal A_{\xi}(t) = |a_{\xi}| e^{-2\pi t A(\xi)} \,, \quad a_\xi =
|a_{\xi}| e^{2\pi i\alpha_{\xi}}
$$
Then we want to give a lower bound for
\begin{equation}\label{LHS of base case}
 \int \left| \mathcal A_{\xi}(t) e^{2\pi i(\alpha_{\xi} + \langle
   \xi-\xi_0,G(t,\^x) \rangle)} +
\mathcal A_{\xi'}(t) e^{2\pi i(\alpha_{\xi'} + \langle \xi'-\xi_0,G(t,\^x)
  \rangle)} \right|^2 d\mu
\end{equation}

We have
\begin{multline}
  \left| \mathcal A_{\xi}(t) e^{2\pi i(\alpha_{\xi} + \langle \xi-\xi_0,G(t,\^x) \rangle)} +
\mathcal A_{\xi'}(t) e^{2\pi i(\alpha_{\xi'} + \langle \xi'-\xi_0,G(t,\^x)
  \rangle)} \right|^2 \\
= \mathcal A_{\xi}(t)^2+ \mathcal A_{\xi'}(t)^2 +
 2\mathcal A_{\xi}(t)\mathcal A_{\xi'}(t)\cos2\pi \phi(t,\^x)
\end{multline}
where the phase function is
\begin{equation}
  \phi(t,\^x) = \alpha_{\xi} - \alpha_{\xi'}  + |\xi-\xi'|
  \Phi_{\xi,\xi'}(t,\^x) \,, \quad \Phi_{\xi,\xi'}(t,\^x) = \langle
  \frac{\xi-\xi'}{|\xi-\xi'|} ,G(t,\^x) \rangle
\end{equation}

Let
$$
\mathcal S_\delta = \{(t,\^x) \in \supp \psi: \cos 2\pi \phi(t,\^x)
\geq 
-1+\delta \}
$$ 
According to \eqref{grad Phi},
\begin{equation}
 | \nabla \phi(t,\^x)| \geq |\xi-\xi'| C\geq C
\end{equation}
for all $(t,\^x)\in \supp \psi$ and all
$\xi\neq \xi'\in \vE'$ (using integrality for
$|\xi-\xi'|\geq 1$). Therefore since the phase varies by at least a
fixed amount, there is some $\delta>0$ and $C >0$
(independent of $\xi,\xi'$) so that
\begin{equation}
  \int_{\mathcal S_\delta} d\mu \geq \frac C\delta
\end{equation}
On the set $\mathcal S_\delta$ we have
\begin{multline*}
\mathcal A_{\xi}(t)^2+ \mathcal A_{\xi'}(t)^2 +
2\mathcal A_{\xi}(t)\mathcal A_{\xi'}(t)\cos2\pi \phi(t,\^x)  \\
\geq  \mathcal A_{\xi}(t)^2+ \mathcal A_{\xi'}(t)^2 +
(-1+\delta) 2\mathcal A_{\xi}(t)\mathcal A_{\xi'}(t) \\
 =
\delta ( \mathcal A_{\xi}(t)^2+ \mathcal A_{\xi'}(t)^2  )
+ (1-\delta) \left| \mathcal A_{\xi}(t) - \mathcal A_{\xi'}(t)\right|^2  \\
 \geq \delta (|a_\xi|^2e^{-8\pi \tau A(\xi)}  + |a_{\xi'}|^2 e^{-8\pi \tau A(\xi')} )
\end{multline*}

Therefore we find that
\begin{equation*}
 \begin{split}
  \eqref{LHS of base case} & =
\int \left\{ \mathcal A_{\xi}(t)^2+ \mathcal A_{\xi'}(t)^2 +
 2\mathcal A_{\xi}(t)\mathcal A_{\xi'}(t)\cos2\pi \phi(t,\^x) \right\}
 d\mu \\
& \geq  \int_{\mathcal S_\delta} \left\{  \mathcal A_{\xi}(t)^2+ \mathcal A_{\xi'}(t)^2 +
 2\mathcal A_{\xi}(t)\mathcal A_{\xi'}(t)\cos 2\pi \phi(t,\^x)
\right\}  d\mu\\
&\geq
\delta (|a_\xi|^2e^{-8\pi\tau A(\xi)}  + |a_{\xi'}|^2 e^{-8\pi\tau A(\xi')} )
\int_{\mathcal S_\delta} d\mu \\
&\geq C( |a_\xi|^2e^{-8\pi\tau A(\xi)}  + |a_{\xi'}|^2 e^{-8\pi\tau A(\xi')} )
 \end{split}
\end{equation*}
as claimed.

The case $d'\geq 2$:
By Proposition~\ref{Separation lemma} we may partition $\vE'=\coprod \vE_\alpha$ where
\begin{equation}
  \diam \vE_\alpha \ll \lambda^{\frac 1{d+1}}\,, \quad
\dist(\vE_\alpha,\vE_\beta) \gg \lambda^{\frac 1{(d+1)(1+c(d))}}
\,, \alpha\neq \beta
\end{equation}
Then
\begin{multline*}
 \int \left| \sum_{\xi\in \vE'} a_\xi e^{2\pi i \langle
   \xi-\xi_0,G(t,\^x) \rangle } e^{-2\pi tA(\xi)} \right|^2 d\mu
\\
= \sum_\alpha \int \left| \sum_{\xi\in \vE_\alpha}
a_\xi e^{2\pi i \langle \xi-\xi_0,G(t,\^x) \rangle } e^{-2\pi tA(\xi)} \right|^2 d\mu
+ \sum_{\alpha\neq \beta} \sum_{\xi\in \vE_\alpha}
\sum_{\xi'\in \vE_\beta} a_\xi \overline{a_{\xi'}} J_{\xi,\xi'}
\end{multline*}
where the oscillatory integral $J_{\xi,\xi'}$ is given in \eqref{def of J}.
By Lemma~\ref{lem:oscillatory integral}  we have
\begin{equation}
 |J_{\xi,\xi'}|  \ll \frac {1} {|\xi-\xi'|^r}
 \ll \frac  1{\lambda^N} 
\,, \quad \forall N>1
\end{equation}
since $|\xi-\xi'|\gg \lambda^{\frac 1{(d+1)(1+c(d))}}$ for $\xi\in
\vE_\alpha$, $\xi'\in \vE_\beta$ with $\alpha\neq \beta$.
Hence we have an upper bound for the off-diagonal terms
\begin{equation}\label{upper bd for off-diag}
 \sum_{\alpha\neq \beta} \sum_{\xi\in \vE_\alpha} \sum_{\xi'\in \vE_\beta} a_\xi \overline{a_{\xi'}} J_{\xi,\xi'} \ll
\frac 1{\lambda^N} \sum_{\xi\in \vE'} |a_\xi|^2   
\leq \frac 1{\lambda^N}\,, \quad \forall N>1
\end{equation}
 taking into account the normalization $\sum_\xi |a_\xi|^2=1$.


We now want to derive a lower bound for the diagonal terms.
By Jarnik's theorem, since $\diam \vE_\alpha \ll \lambda^{\frac
  1{d+1}}$, the set $\vE_\alpha$ is contained in an affine hyperplane
$H_\alpha$ and hence $d(\vE_\alpha)\leq d-1$. Thus by the induction hypothesis we have
\begin{equation}
 \int \left| \sum_{\xi\in \vE_\alpha} a_\xi e^{2\pi i \langle \xi-\xi_0,G(t,\^x) \rangle }
 e^{-2\pi tA(\xi)} \right|^2 d\mu
\geq C \sum_{\xi \in \vE_\alpha} |a_\xi|^2 e^{-8\pi \tau A(\xi)}
+ O(\frac 1{\lambda^N})
\end{equation}
Combining with the upper bound \eqref{upper bd for off-diag} for the
off-diagonal terms  and summing over $\alpha$ we get the required
lower bound \eqref{inductive assumption}.
\end{proof}

\newpage
\section{Non-zero curvature: Proof of Theorem~\ref{thm:nonzero curvature}}\label{sec:all freqs}
Assume now that the hypersurface $\surface$ has nowhere zero Gauss-Kronecker
curvature, i.e. all principal curvatures  are nowhere zero. (The
condition that $\surface$ is not flat means that at least one of the
principal curvatures is nonzero).
We will use Theorem~\ref{lem:small a's} to prove
Theorem~\ref{thm:nonzero curvature}, namely
that for $\lambda>\lambda_{\surface}$, an eigenfunction
$\varphi_\lambda$ cannot vanish identically on $\surface$.


\subsection{The strategy}
We keep the normalization $\sum_\xi |a_\xi|^2=1$.
We assume that there is some cap
\begin{equation}
\Omega_0 = \Ccap(w_0,\theta_0) \subset S^{d-1}
\end{equation}
around $w_0$, with opening angle $\theta$ of size $O(1)$ so that
\begin{equation}
 |a_\xi|<\frac 1{\lambda^N}, \quad \forall \frac{\xi}{|\xi|} \in \Omega_0
\end{equation}
guaranteed by Theorem~\ref{lem:small a's}. We shall call such frequencies ``negligible''.

We aim to show that there is a {\em larger} cap
$$\Omega_1 = \Ccap(w_1,\theta_1) \subset S^{d-1}$$
for which  all frequencies $\vE_1:= \lambda\Omega_1\cap \vE$
in direction $\Omega_1$ are negligible. Here ``larger'' means
that say
\begin{equation}\label{bigness of Omega_1}
\theta_1\geq \theta_0+\delta_0
\end{equation}
for some fixed $\delta_0>0$ (independent of $\lambda$).
We will show that all frequencies in $\vE_1$ are ``negligible''.
Proceeding in this way we will eventually show that {\em all} frequencies
are ``negligible'', contradicting $\sum_\xi|a_\xi|^2=1$.

\subsection{An oscillatory integral}\label{sec:osc integral 2}

Since $\surface$ has non-vanishing curvature, the unit normals to
$\surface$ sweep out at least a cap  $\Ccap(u_0,\delta_1)$ for some
$\delta_1>0$. If $\delta_0<\delta_1/2$,
then for any $u\in \Ccap(u_0,\frac {\delta_1}2)$ we have
\begin{equation}
 \Ccap(u,\delta_0) \subset \Ccap(u_0,\delta_1)
\end{equation}
We choose $\delta_0$ sufficiently small so that for any such $u$,
we have a patch $\surface_u\subset \surface$ on the surface so that
the Gauss map
\begin{equation}\label{diffeo}
N:\surface_u\to \Ccap(u,\delta_0)
\end{equation}
is a diffeomorphism.

Fix a bump function $\psi$ supported in the cap $\Ccap(u_0,\delta_0)$,
from which we get a smooth measure $d\mu$ on the cap; applying
rotations give smooth measure on any cap $\Ccap(u,\delta_0)$, and
pulling back to the patch $\surface_u$ via the Gauss map $N$ we get a
smooth measure $\mu_u$ on $\surface_u$, which depends in a bounded way
on $u\in \Ccap(u_0,\delta_1/2)$. Denote by
\begin{equation}
\^\mu_u(\xi):=\int_{\surface_u} e^{-2\pi i \langle \xi,x\rangle}  d\mu_u(x)
\end{equation}
its Fourier transform.

Fourier transforms of surface-carried measure are known to decay polynomially in the presence or curvature \cite{Hlawka, Herz},
in fact if the surface is not flat \cite{Littman}. 
However there is faster decay in directions which are disjoint from the image of the Gauss map. We use this to prove:
\begin{lemma}\label{nonstationary phase lemma}
For all vectors $  y\neq 0$ which do not lie in the direction of the
bigger cap $\Ccap(u,2\delta_0)$, we have
\begin{equation}
\^\mu(y)  \ll_{N} \frac
 1{|y|^N},\quad \frac{  y}{|  y|} \notin \Ccap(u,2\delta_0),
 \quad \forall N>0
\end{equation}
where the implied constants can be taken uniform in $u\in \Ccap(u_0,\delta_1/2)$
\end{lemma}
\begin{proof}
 We take a regular parametrization $X:t=(t_1,\dots, t_{d-1})\mapsto
 X(t)$ of the patch $\surface_u$. Then the Fourier transform becomes
 \begin{equation*}
   \^\mu_u(y) = \int_{\R^{d-1}} e^{-2\pi i |y| \phi(t)} \Psi(t) dt
 \end{equation*}
for a suitable amplitude $\Psi\in C_c^\infty(\R^{d-1})$ and with phase function
\begin{equation*}
  \phi(t) = \langle \^y, X(t) \rangle,\quad \^y:=\frac{y}{|y|}
\end{equation*}
Our claim will follow by integration by parts
if we give a uniform lower bound for the
gradient of the phase function
\begin{equation}
  |\nabla \phi| \geq C>0
\end{equation}

The gradient of the phase function is given by
\begin{equation*}
  \nabla \phi = DX(t) \^y
\end{equation*}
where we think of $\^y\in S^{d-1}$ as a column vector and the
derivative $DX=(\frac{\partial X_i}{\partial t_j})$ is a $(d-1)\times
d$ matrix.
Choose a row vector $\omega\in \R^{d-1}$ for which $\omega DX(t)$ is
the orthogonal projection $P(\^y)$ of $\^y$ on the tangent space
$T_{X(t)}\surface$. Thus
\begin{equation*}
\omega \nabla \phi(t)  = \omega DX(t) \^y =  P(\^y)\cdot \^y
\end{equation*}

Let $\alpha$ be the angle between the unit normal
$N_{X(t)}$ and $\^y$. Then $\alpha>\delta_0$ since
by assumption $N_{X(t)}\in \Ccap(u,\delta_0)$
while $\^y\notin \Ccap(u,2\delta_0)$. Therefore
\begin{equation*}
   |P(\^y)| = |\^y| \sin \alpha = \sin \alpha \geq \sin \delta_0
\end{equation*}
and therefore
\begin{equation*}
\omega\nabla \phi(t)=  P(\^y)\cdot \^y= |P(\^y)|^2\geq (\sin \delta_0)^2
\end{equation*}
On the other hand,
\begin{equation*}
  |\omega \nabla \phi(t)| \leq |\omega| |\nabla \phi(t)|
\end{equation*}
and so we find
\begin{equation*}
  |\nabla \phi|\geq \frac{(\sin \delta_0)^2}{|\omega|}
\end{equation*}
and it remains to give an upper bound for $|\omega|$.

We have
\begin{equation*}
|P\^y|^2=  |\omega DX|^2  = \omega DX  DX^T \omega^T
\end{equation*}
Now $P(\^y)$ has length at most $|\^y|=1$, being the orthogonal
projection of the unit vector $\^y$, and so we find
\begin{equation*}
  1 \geq  \omega DX(t)  DX(t)^T \omega^T
\end{equation*}
The rows of $DX(t)$ are linearly independent since we assume that $X$
is a regular parametrization. Hence the quadratic form $DX(t)
DX(t)^T$ is positive definite  (it is the first fundamental form of
the hypersurface) and so
\begin{equation*}
  \omega DX(t)  DX(t)^T \omega^T \geq c(t) |\omega|^2
\end{equation*}
for some $c(t)>0$, and taking $c:=\min\{ c(t)\}>0$ we  find
$1 \geq c |\omega|^2$, that is
\begin{equation*}
  |\omega|\leq \frac 1{\sqrt{c}}
\end{equation*}
Thus
\begin{equation}
  |\nabla \phi|\geq \frac{(\sin \delta_0)^2}{\sqrt{c}}
\end{equation}
giving the required lower bound.
\end{proof}

\subsection{Geometric considerations}
For a unit vector $u\in S^{d-1}$ let
\begin{equation}
 \tau_u(x) = x- 2\langle x,u \rangle  u
\end{equation}
be the reflection in the hyperplane orthogonal to $u$.

Fix $u_0\in S^{d-1}$, and $\delta>0$. Then there is some
$\epsilon=\epsilon_d(\delta)>0$ so that for every $w\in S^{d-1}$, the
set of reflected points $\tau_u w$, for $u$ ranging over all points in
the cap $\Ccap(u_0,\delta)$, contains a cap $\Ccap(w_1,\epsilon)$:
\begin{equation}\label{def of epsilon}
 \forall w \quad \exists w_1 \mbox{  such that  } \quad
 \Ccap(w_1,\epsilon) \subseteq \{\tau_u w: u\in \Ccap(u_0,\delta) \}
\end{equation}
By symmetry, $\epsilon$ is independent of the base point $u_0$, and
depends only on the dimension $d$ and on $\delta$.


Now let $u_0\in S^{d-1}$, $\delta_1>0$ be as in \S~\ref{sec:osc integral 2}.
We fix  $\delta_0>0$, with $\delta_0<\delta_1/2$ sufficiently small so
that the Gauss map gives a diffeomorphism \eqref{diffeo}, and
Lemma~\ref{nonstationary phase lemma} holds.
In addition we require
\begin{equation}
  \delta_0 <  \frac 16\epsilon_d(\frac{\delta_1}2)
\end{equation}


Recall that $\Ccap(w_0,\theta_0)=\Omega_0$ is a cap where we assume
the Fourier coefficients $a_\xi\approx 0$
are negligible for all frequencies with $\xi/|\xi|\in \Omega_0$.

\begin{lemma}\label{lem:playing with caps}
Let $u\in \Ccap(u_0,\delta_1/2)$ and
$B \subset \tau_u \Ccap(w_0,\theta_0-4\delta_0)$.
Then for all unit vectors $y\notin B$, either $y\in \Omega_0 =
\Ccap(w_0,\theta_0)$ or else
\begin{equation}\label{nonresonance u}
 \frac{x-y}{|x-y|} \notin \Ccap(u,2\delta_0), \quad \forall x\in
 \tau_u \Ccap(w_0,\theta_0-4\delta_0), x\neq y
\end{equation}
\end{lemma}
\begin{proof}
Let $y\in S^{d-1}\backslash B$ and assume that \eqref{nonresonance u} fails,
that is there is some $x\in \tau_u \Ccap(w_0,\theta_0-4\delta_0)$ and
$u_1\in \Ccap(u,2\delta_0)$ so that
\begin{equation}\label{prelim reflection}
 y=x-|x-y|u_1
\end{equation}
We then need to show that $y\in \Omega_0$.

The condition \eqref{prelim reflection} means that
\begin{equation}
 y=\tau_{u_1}(x)
\end{equation}
This is because $y$ lies on the intersection of the sphere with the line through $x\in S^{d-1}$
in the direction of $u_1$; that intersection contains (at most) two points, one of them being $x$,
which we assume is distinct from $y$.
Clearly the reflection $\tau_{u_1}(x)$ also has this properties, so that $y=\tau_{u_1}(x)$.
Hence we find that
\begin{equation*}
 y=\tau_{u_1}x \in \tau_{u_1} \circ \tau_u \Ccap(w_0,\theta_0-4\delta_0)
\end{equation*}

The composition of two distinct reflections $\tau_u \circ\tau_{u_1}$ is a rotation in the plane spanned by the two vectors $u,u_1$ (assumed not to be co-linear) by an angle which is twice the angle $\alpha$ between the two vectors. In our case, since
$u_1\in \Ccap(u,2\delta_0)$  lies in  cap  centered at  $u $, we have $\alpha\leq 2\delta_0 $. Hence
\begin{equation*}
 \tau_{u_1} \circ \tau_u \Ccap(w_0,\theta_0-4\delta_0)\subseteq \Ccap(w_0,\theta_0-4\delta_0+2\alpha)
\subseteq \Ccap(w_0,\theta_0)
\end{equation*}
so that $y\in \Omega_0$ as claimed.
\end{proof}

\subsection{Vanishing of Fourier coefficients in the larger caps}

For each $u\in \Ccap(u_0,\delta_1/2)$, apply the reflection $\tau_u$ to the cap $\Ccap(w_0, \theta_0-5\delta_0)$ to get a cap
\begin{equation}
 \Omega_1(u):=\tau_u \Ccap(w_0,\theta_0-5\delta_0) = \Ccap(\tau_u w_0,\theta_0-5\delta_0)
\end{equation}
We now claim that the Fourier coefficients $a_\xi$ for frequencies whose directions lie
in the cap $\Omega_1(u)$ are negligible:
\begin{proposition}
If $\xi/|\xi|\in \Omega_1(u)$ then
 \begin{equation}\label{upper bd for a}
| a_\xi | \ll  \frac 1{\lambda^N},\quad \forall N\geq 1
\end{equation}
\end{proposition}
\begin{proof}
 Let
\begin{equation}
 \vF = \lambda   \Omega_1(u) \cap \vE
\end{equation}
and use Proposition~\ref{Separation lemma} with $\rho=\lambda^{\delta(d)}$ to get an overset $\tilde \vF$, $\vF\subseteq \tilde\vF \subset \vE$ so that
\begin{equation}
 \dist(\tilde \vF,\vE\backslash \tilde \vF) >\lambda^{\delta(d)}
\end{equation}
and
\begin{equation}\label{eq:7:17}
 \diam \tilde \vF \leq \diam \vF  + \lambda^{\frac 1{2(d+1)}} \leq  \lambda \diam \Omega_1(u)\{ 1+O(\lambda^{-1+\frac 1{2(d+1)}})\}
\end{equation}
Since $\frac 1\lambda \tilde \vF \supset \frac 1\lambda \vF \subset \Omega_1(u)= \tau_u\Ccap(w_0,\theta_0-5\delta_0)$,
condition \eqref{eq:7:17} implies that for $\lambda$ sufficiently large,
\begin{equation}
 \frac 1\lambda \tilde \vF \subset \tau_u \Ccap(w_0,\theta_0-4\delta_0)
\end{equation}

Set
\begin{equation}
 \vE_1(u) = \tilde \vF
\end{equation}
so that
\begin{equation}
 \dist(\vE_1(u),\vE\backslash \vE_1(u) )>\lambda^{\delta(d)}
\end{equation}
and
\begin{equation}
 \frac 1\lambda \vE_1(u) \subset \tau_u\Ccap(w_0,\theta_0-4\delta_0)
\end{equation}

Consider the integral
\begin{equation}
  0=\int_{\surface_u} \overline{\varphi(x)} \sum_{\xi\in \vE_1(u)} a_\xi e^{2\pi i \langle \xi,x \rangle} d\mu_u(x)
\end{equation}
which equals zero, since we assume $\varphi=0$ on $\surface$.

On the other hand, expanding
$$\varphi = \sum_{\xi\in \vE_1(u)} a_\xi e^{2\pi i \langle \xi,x \rangle}+\sum_{\xi\notin \vE_1(u)}
a_\xi e^{2\pi i \langle \xi,x \rangle}
$$
gives a sum of ``diagonal'' and ``off-diagonal'' terms:
\begin{equation*}
 \begin{split}
0 &= \int_{\surface_u} \left|\sum_{\xi \in \vE_1(u)} a_\xi e^{2\pi i \langle \xi,x \rangle}\right|^2 d\mu_u(x) +
\sum_{\xi\in \vE_1(u)}\sum_{\eta\notin \vE_1(u)} a_\xi\overline{a_\eta} \^\mu_u(\eta-\xi)\\
 &= \mbox{diagonal }+ \mbox{off-diagonal}
 \end{split}
\end{equation*}
The diagonal term can be bounded from below:
\begin{equation}
  \int_{\surface_u} \left|\sum_{\xi \in \vE_1(u)} a_\xi e^{2\pi i \langle \xi,x \rangle}\right|^2 d\mu_u(x)
  \geq C  \sum_{\xi\in \vE_1(u)} |a_\xi|^2
\end{equation}
by arguing as in Proposition~\ref{lem:lower bd for short mean square}
(in fact by using it in the special case $A(\xi)=0$).

We will show that the off-diagonal part is ``negligible'' which will
give the required upper bound \eqref{upper bd for a}.
To do so, decompose the off-diagonal term as
\begin{equation*}
\begin{split}
\mbox{off-diagonal} &= \sum_{\xi\in \vE_1(u)} \sum_{\eta\in
  \vE_0\backslash \vE_1(u)}
a_\xi\overline{a_\eta} \^\mu_u(\eta-\xi) \\ &+
\sum_{\xi\in \vE_1(u)} \sum_{\eta\in \vE\backslash (\vE_0\cup \vE_1(u))}
a_\xi\overline{a_\eta} \^\mu_u(\eta-\xi)
\end{split}
\end{equation*}
The first term is negligible because all the coefficients
$a_\eta\approx 0$ are negligible for $\eta\in \vE_0$.

In the second term, we claim that all Fourier
transforms $\^\mu_u(\xi-\eta)\approx 0$ are negligible for $\xi\in
\vE_1(u)$, $\eta\notin \vE_1(u)\cup \vE_0$: Indeed, denoting by $x=\xi/|\xi|$ and $y=\eta/|\eta|$ (note $x\neq y$),
apply Lemma~\ref{lem:playing with caps} with
$$B=\frac 1\lambda \vE_1(u) \subset \tau_u \Ccap(w_0, \theta_0-4\delta_0)$$
Then $y\notin \Omega_0$ since $\eta\notin \vE_0$ hence
$$\frac{x-y}{|x-y|} \notin \Ccap(u,2\delta_0)
$$
Since $|\xi|=|\eta|=\lambda$, we have $\frac{\xi-\eta}{|\xi-\eta|} = \frac{x-y}{|x-y|}$. Therefore
$$
\frac{\xi-\eta}{|\xi-\eta|} = \frac{x-y}{|x-y|} \notin \Ccap(u,2\delta_0)
$$
and hence by the non-stationary phase
lemma~\ref{nonstationary phase lemma}, we have
$$\^\mu_u(\eta-\xi) \ll \frac 1{|\xi-\eta|^M},\quad \forall M\geq 1$$
Moreover, since $\xi\in \vE_1(u)$ and $\eta\notin \vE_1(u)$,
$$
|\xi-\eta|\geq \dist(\vE_1(u),\vE\backslash \vE_1(u)) >\lambda^{\delta(d)}
$$
Hence we get
\begin{equation}
  \^\mu(\eta-\xi) \ll \frac 1{\lambda^N},\quad \xi\in \vE_1(u), \eta
  \notin \vE_1(u)\cup \vE_0
\end{equation}
that is the Fourier transforms are negligible as required. Thus the off-diagonal term is negligible, which shows that
$\sum_{\xi\in \vE_1(u)} |a_\xi|^2$ is negligible. Since $\frac 1\lambda \vF = \Omega_1(u)\cap \frac 1\lambda \vE\subset \frac 1\lambda \vE_1(u)$, we get $|a_\xi|\ll \frac 1{\lambda^N}$ if $\xi/|\xi|\in \Omega_1(u)$.
\end{proof}

Finally, we claim
\begin{proposition}
There is a cap $\Omega_1 = \Ccap(w_1,\theta_0+\delta_0)$ for which all
frequencies $\xi$ in direction $\Omega_1$, the Fourier coefficients $a_\xi$ are negligible
\end{proposition}
\begin{proof}
We note that the union
\begin{equation}
 \bigcup_{u\in \Ccap(u_0,\delta_1/2)} \Omega_1(u) = \bigcup_{u\in \Ccap(u_0,\delta_1/2)} \Ccap(\tau_u w_0, \theta_0-5\delta_0)
\end{equation}
contains a cap $\Omega_1 = \Ccap(w_1,\theta_1)$ with $\theta_1\geq \theta_0 +\delta_0$.

This follows since the set of reflected centers
$$\{\tau_u w_0: u \in \Ccap(u_0,\delta_1/2) \}$$
contains a cap $\Ccap(w_1, \epsilon_d(\frac {\delta_1}2))$, where $\epsilon_d(\delta)$ is defined in \eqref{def of epsilon},
since we chose $\delta_0$ sufficiently small so that $\epsilon_d(\frac{\delta_1}2) >6\delta_0$, and hence
\begin{equation}
 \bigcup_{u\in \Ccap(u_0,\delta_1/2)} \Omega_1(u) \supset \Ccap(w_1,\theta_0-5\delta_0+\epsilon_d(\frac {\delta_1}2))
\supset \Ccap(w_1,\theta_0+\delta_0)
\end{equation}
Therefore for all frequencies in direction $\Omega_1$ the Fourier coefficients $a_\xi$ are negligible, since the same holds for each of the small caps $\Omega_1(u)$ containing $\Omega_1$.
\end{proof}
 By continuing this process, we see that {\em all} coefficients $a_\xi$ are negligible,
contradicting the normalization $\sum_\xi |a_\xi|^2=1$. This concludes the proof of Theorem~\ref{thm:nonzero curvature}.

\newpage

\appendix\section{The two-dimensional case: Using the ABC theorem}\label{sec:abc}

In this section we give a proof of Theorem~\ref{thm:2 dim} using the function-field ``abc theorem'' of
Brownawell-Masser \cite{BM} and Voloch \cite{Voloch}. We recall the statement: Let $K=\C(X)$ be the function field of an algebraic curve of genus $g$ over the complex numbers,
$S$ a finite set of places of $K$, and $u_1,\dots u_m\in K$ a set of $S$-units, that is rational functions whose zeros and poles lie in $S$.
The degree, or height, of a non-constant rational function $x\in K$ is defined as the degree of $K$ over the field extension $\C(x)$:
$H(x)=[K:\C(x)]$, which is the number of zeros (or poles) of $x$, counted with multiplicities.
\begin{theorem}[\cite{BM}, \cite{Voloch}]\label{abc theorem}
Let $u_1,\dots, u_m\in K$ be non-constant $S$-units, linearly independent over $\C$, satisfying
\begin{equation}
\sum_{j=1}^m u_j = 1
\end{equation}
Then
\begin{equation}
\max_j H(u_j) \leq \frac{m(m-1)}2 (2g-2+\#S)
\end{equation}
\end{theorem}
This result improves that of R.C. Mason \cite{Mason several},
where the quadratic term $m(m-1)/2$ is replaced by a term exponential in $m$,
which is not sufficiently strong for our purposes.

\subsection{Complexification}
Let $\varphi$ be an eigenfunction, $-\Delta \varphi = 4\pi^2 \lambda^2 \varphi$, which vanishes on the curve $\surface$. Write
\begin{equation}\label{expansion of phi}
\varphi(x) = \sum_{\xi} a_\xi e^{2\pi i \langle \xi,x \rangle}
\end{equation}
Let
\begin{equation}
 \supp \^\varphi = \{\xi: a_\xi\neq 0\}
\end{equation}
be the set of frequencies of $\varphi$, and set
\begin{equation}
 r=\#\supp\^\varphi
\end{equation}
to be the number of frequencies; necessarily $r\geq 2$.

 We can embed the torus $\TT^2\simeq S^1\times S^1$
in $\C^2$ via the map $(x,y)\mapsto (z_1,z_2) = (e^{2\pi i x}, e^{2\pi i y})$. This allows us to
associate with each trigonometric polynomial \eqref{expansion of phi}
a Laurent polynomial 
\begin{equation}\label{sum for F}
 F(z) = \sum_\xi a_\xi z^\xi
\end{equation}
where for $z=(z_1,z_2)\in \C^2$ and $\xi=(n_1,n_2)\in \Z^2$ we denote
$$z^\xi:=z_1^{n_1} z_2^{n_2}$$
We can further write
\begin{equation}
  F(z_1,z_2) = \frac {P(z_1,z_2)}{z_1^{a_1} z_2^{a_2}}
\end{equation}
for a unique polynomial $P\in \C[z_1,z_2]$ so that $z_1\nmid P$, $z_2\nmid P$.
If $\varphi \neq ae^{2\pi i \langle \xi, x \rangle}$ is not composed of a single frequency
(which it cannot if we assume that it is real-valued) then $P$ is non-constant.
Thus to each trigonometric polynomial $\varphi$ (not a single exponential) we associate the plane curve
$$
X_P=\{z: P(z)=0\}\subset \C^2
$$
which is possibly reducible and singular.

The nodal set of $\varphi$ must be contained in $X_P$, since $z_i$ do not vanish on $\TT^2 = S^1\times S^1 \subset \C^2$.
Thus if $\varphi$ vanishes on the (real) smooth curve $\surface\subset \TT^2$,
then $\surface$ must be contained in an irreducible component of $X_P$ (possibly in more than one component).
Thus we get an irreducible component (possibly singular)
$$ X_D=\{z: D(z)=0\}$$
containing $\surface$. Here $D\in \C[z_1,z_2]$ is an irreducible divisor
of $P$ of positive degree.
Note that in that case $D(z_1,z_2)$ cannot depend only on one of
the variables, say on $z_1$. Indeed in that case $D(z_1)$ is a
one-variable polynomial, and is then irreducible only in the case that
it is linear: $D(z_1)=z_1-c$, whose zero set is a closed geodesic, contradicting our choice of $\surface$.

Let $\lambda_0$ be minimal where an eigenfunction $\varphi_0$ with
eigenvalue $4\pi^2\lambda_0^2$ vanishes on $\surface$.
We can choose an irreducible $D_0\in \C[z_1,z_2]$ so that $\surface
\subset X_{D_0}=:X_0$.
Let $ \C(X_0)$ be the function field of the curve $X_0$, that is the
field of fractions
of the integer domain $\C[z_1,z_2]/(D_0)$.  The curve $X_0$ is
irreducible but possibly singular.
Let $X\to X_0$ be its normalization, whose ring of regular functions
is the integral closure of $\C[z_1,z_2]/(D_0)$, and has the same
function field as $X_0$. The map $X\to X_0$ is one-to-one outside of
finitely many points.

Restricting the monomials $z^\xi$ ($\xi\in \Z^2$) to $X_0$ gives rational
functions which we still denote by $z^\xi$,
in $\C(X_0)=\C(X)$, which have all their zeros and poles in the set
$S_0$ given by
\begin{equation}
S_0=X_0\cap \{ (z_1,z_2) \in \mathbb P^2:  z_i=0, \infty\}
\end{equation}
Note that $S_0$ is finite because $X_0$ is not a line of the form
$z_i=0,\infty$, since $D(z_1,z_2)$ depends on both variables.
By pulling back to the normalization $X$,
we get rational functions, still denoted by $z^\xi$, on $X$ which are
$S$-units for the pullback $S$ of $S_0$ to $X$.

\subsection{A lower bound for the height of monomials}
In order to apply Theorem~\ref{abc theorem}, we need to compute the height of the monomials $z^\xi$ as
rational functions on the curve $X$.
The assumption that $\surface$ is not a segment of a closed geodesic allows us to obtain a useful lower bound:

\begin{lemma}\label{lem:ineq for height}
Suppose that $\surface$ is not a segment of a closed geodesic.
Then there is some constant $c_\surface>0$ so that
$$H(z_1^{n_1} z_2^{n_2}) \geq c_\surface \max(|n_1|, |n_2|)$$
for all  $(n_1,n_2)\in \Z^2$.
\end{lemma}
\begin{proof}

Let $\Div X$ be the vector space of divisors of $X$, that is of (finite) formal sums $\sum_{P\in X} n_P P$
(we include points at infinity). The degree of such a divisor is $\sum_P n_p$.
For a rational function $u$ on $X$, we have an associated principal divisor $\Div u=\Div_0 u-\Div_\infty u$ where $\Div_0 u$ and $\Div_\infty u$ are the divisors of zeros and poles. Then the degree of a principal divisor is zero: $ \deg \Div u = \deg \Div_0 u - \deg \Div_\infty u = 0$ and the height  of $u$ equals
$$ H(u) = \deg \Div_0 u=\deg \Div_\infty u$$

On the vector space $\Div X$ we have the $\ell^1$-norm
$$
||\sum_P n_P P||_1:=\sum_P |n_P|
$$
which for a principal divisor equals twice the height of $u$.
\begin{equation*}
 ||\Div u||_1 = \deg\Div_0 u +\deg \Div_\infty u = 2H(u)
\end{equation*}

We  claim that if $\surface$ is not a segment of a closed geodesic,
then  $\Div z_1$ and $\Div z_2$ are linearly independent elements of $\Div X$.
Indeed, a linear dependence means that there are integers $a_1,a_2\in \Z$ for which
$$ a_1 \Div z_1 = a_2 \Div z_2$$
or equivalently that $z_1^{a_1}/z_2^{a_2}\equiv c$ is the constant function when restricted to the curve $X_0$. That means that on $\surface$, we have
$$ e^{2\pi i(a_1 x_1-a_2x_2)}=c$$
whose zero set is a union of closed geodesics. Hence $\Div z_1$ and $\Div z_2$ are linearly independent.

Since $\Div z_1$ and $\Div z_2$ are linearly independent, their span $V$ is a two-dimensional vector space in $\Div X$. On $V$ we then have two norms: The restriction of the $\ell^1$-norm and the $\ell^\infty$-norm
\begin{equation*}
 || \Div (z_1^{n_1} z_2^{n_2})||_\infty = ||n_1 \Div z_1 + n_2 \Div z_2||_\infty = \max(|n_1|,|n_2|)
\end{equation*}
which is indeed a norm since $\Div z_1$ and $\Div z_2$ are linearly independent.
Since on any finite-dimensional vector space all norms are equivalent, we find that there is some $c=c_V>0$ for which
\begin{equation*}
      || \Div (z_1^{n_1}z_2^{n_2} )||_1
\geq c  | \Div (z_1^{n_1}z_2^{n_2})||_\infty = c \max(|n_1|,|n_2|)
\end{equation*}
for all $n\in \Z^2$, and hence
\begin{equation*}
 H(z_1^{n_1}z_2^{n_2}) =\frac 12 || \Div (z_1^{n_1}z_2^{n_2} )||_1\geq  \frac 12 c  \max (|n_1|,|n_2|)
\end{equation*}
as claimed.
\end{proof}

\subsection{Proof of Theorem~\ref{nonvanishing thm}}
We assume that $\surface$ is not a segment of a closed geodesic.
We choose $\lambda$ sufficiently large so that
\begin{equation}
 \lambda \gg  (\#S +2g_X-2)^{1+\epsilon}
\end{equation}
and show that no eigenfunction $\varphi$ with eigenvalue
$4\pi^2\lambda^2$ can vanish on $\surface$.

Suppose $\lambda$ admits an eigenfunction \eqref{expansion of phi}
which vanishes on $\surface$. Among such eigenfunctions,
choose such $\varphi$ with the number of frequencies $r$ being minimal.
If $r=2$ then after scaling,
$$
\varphi(x) = e^{2\pi i \langle \xi,x\rangle}  -ae^{2\pi i \langle \xi',x\rangle} , \qquad a\in \C
$$
and for its nodal set to contain a real point,  we need $|a|=1$, that is
$a=e^{2\pi i\alpha}$, $\alpha\in \R$.
In that case the nodal set consists of $x\in \TT^2$ with
$$ \langle \xi-\xi', x\rangle  \in \alpha +\Z$$
which is a union of straight lines with rational slopes, i.e.
closed geodesics. So we may assume $r\geq 3$.

In the expansion \eqref{expansion of phi}, choose one of the
frequencies $\xi_0$ and divide all terms in
\eqref{expansion of phi} by $a_{\xi_0} e^{2\pi i \langle \xi_0,x \rangle}$
to get a relation:
\begin{equation}\label{divided relation}
\sum_{\xi_0\neq \xi \in \supp \^\varphi} -\frac{a_\xi}{a_{\xi_0}}z^{\xi-\xi_0}=1
\end{equation}
Set
\begin{equation}
 u_\xi:= -\frac{a_\xi}{a_{\xi_0}}z^{\xi-\xi_0} \in \C(X_0)
\end{equation}
Then we get a relation in $\C(X_0)=\C(X)$ (an $S$-unit equation)
\begin{equation}\label{S-unit eq from phi}
\sum_{\substack{\xi\in \supp \^\varphi\\\xi\neq \xi_0}}  u_\xi   = 1
\end{equation}
where $u_\xi$ are linearly independent, by the minimality assumption on $\varphi$.
To the relation \eqref{S-unit eq from phi} we apply the ``abc-Theorem'' (Theorem~\ref{abc theorem})
which says that if $r\geq 3$
then
\begin{equation}\label{Voloch thm}
 \max ( H(u_\xi): \xi_0\neq \xi\in \supp \^\varphi ) \leq \frac {(r-1)(r-2)}2
 (\#S +2g_X-2)
\end{equation}
where $g_X$ is the genus of the smooth curve $X$.
Since
$$H( u_\xi ) =H( z^{\xi-\xi_0}) \geq c_\surface ||\xi-\xi_0||_\infty$$
by Lemma~\ref{lem:ineq for height}, we find that
\begin{equation}\label{conseq of abc}
\max_{\xi\neq \xi_0} ||\xi-\xi_0||_\infty \ll  \frac {(r-1)(r-2)}2 (\#S +2g_X-2)
\end{equation}

Now the number of frequencies $r$ is at most the total number of lattice points on the circle $|x|=\lambda$,
hence is bounded by $r\ll \lambda^\epsilon$
for all $\epsilon>0$. Thus by \eqref{conseq of abc} we find that all frequencies  of
$\varphi$ are contained in a box of size $\ll \lambda^\epsilon$
around $\xi_0$.
By Jarnik's theorem (Theorem~\ref{Jarnik's thm}), any arc of size $\ll \lambda^{1/3}$ contains at
most two lattice points, hence this forces $r=2$ contradicting our
assumption $r\geq 3$.
This gives a contradiction for $\lambda$ sufficiently large. \qed

\end{document}